\documentclass[sn-mathphys-num]{sn-jnl}
\usepackage[utf8]{inputenc}
\usepackage[english]{babel}
\usepackage{hyperref}
\usepackage{textcomp}
\usepackage[autostyle]{csquotes}
\usepackage{dsfont}
\usepackage{graphicx}%
\usepackage{multirow}%
\usepackage{amsmath,amssymb,amsfonts}%
\usepackage{amsthm}%
\usepackage{mathrsfs}%
\usepackage[title]{appendix}%
\usepackage{xcolor}%
\usepackage{textcomp}%
\usepackage{lmodern}
\usepackage{kpfonts}
\usepackage{comment}
\usepackage{notations}
\usepackage{tikz}
\usepackage{pgfplots}
\pgfplotsset{/pgf/number format/use comma,compat=newest}

\newcommand\iid{\mathrel{\stackrel{\makebox[0pt]{\mbox{\normalfont\tiny iid}}}{\sim}}}

\theoremstyle{plain} 
\newtheorem{theorem}{Theorem}

\newtheorem{lemma}[theorem]{Lemma} 
\newtheorem{proposition}[theorem]{Proposition} 

\theoremstyle{definition} 
\newtheorem{definition}{Definition}

\theoremstyle{remark} 
\newtheorem{remark}{Remark}

\begin{document}

\title[Article Title]{On the phase diagram of the multiscale mean-field spin-glass}

\author*[1]{\fnm{Francesco} \sur{Camilli}}\email{fcamilli@ictp.it}

\author[2]{\fnm{Pierluigi} \sur{Contucci}}\email{pierluigi.contucci@unibo.it}

\author[2]{\fnm{Emanuele} \sur{Mingione}}\email{emanuele.mingione2@unibo.it}

\author[2]{\fnm{Daniele} \sur{Tantari}}\email{daniele.tantari@unibo.it}
\equalcont{Authors are ordered in alphabetical order and they all contributed equally.}

\affil[1]{\orgdiv{Quantitative Life Sciences}, \orgname{International Centre for Theoretical Physics}, \orgaddress{\street{Str. Costiera 11}, \city{Trieste}, \postcode{34151}, \country{Italy}}}

\affil[2]{\orgdiv{Dipartimento di Matematica}, \orgname{Alma Mater Studiorum - Università di Bologna}, \orgaddress{\street{Via Zamboni 33}, \city{Bologna}, \postcode{40126}, \country{Italy}}}

\abstract{In this paper we study the phase diagram of a Sherrington-Kirkpatrick (SK) model where the couplings are forced to thermalize at different time scales. Besides being a challenging generalization of the SK model, such settings may arise naturally in physics whenever part of the many degrees of freedom of a system relaxes to equilibrium considerably faster than the others. For this model we
compute the asymptotic value of the second moment of the overlap distribution.
Furthermore, we provide a rigorous sufficient condition for an annealed solution to hold, identifying a high temperature, {or weak coupling}, region. In addition, we also prove that for sufficiently {strong couplings} the solution must present a number of replica symmetry breaking levels at least equal to the number of time scales already present in the multiscale model. Finally, we give a sufficient condition for the existence of gaps in the support of the functional order parameters.
}



\maketitle

\section{Introduction}\label{sec:Intro}

The Sherrington-Kirkpatrick (SK) model \cite{SKoriginal,MPV,Tala_vol1,Tala_vol2,Panchenko2013,Parisi80} is a spin system with random interactions characterized by two main features. Firstly the interaction is of mean field nature. Secondly the randomness is \textit{quenched}, i.e.\ the interaction couplings are random parameters and not thermodynamic degrees of freedom. This means that the thermodynamics is described by the so called \textit{quenched measure}: given a realization of the disorder the spins thermalize according to the Boltzmann-Gibbs distribution,  and the interaction randomness is treated with a successive statistical average. 

In the \textit{multiscale} SK model \cite{contucci-mingione} part of the random interactions become thermodynamic degrees of freedom and the quenched measure is replaced by a \textit{multiscale measure}, obtained by the following procedure. The interactions are divided in a finite number of families equilibrating in a hierarchical succession. Recursively, each family thermalizes at increasing temperature  according to a Boltzmann-Gibbs distribution, where the effective Hamiltonian is given by the free energy associated to the equilibrium of the previous  family. The above recursive construction is deeply related to Derrida-Ruelle probability cascades \cite{Derrida,Ruelle1987,Bolthausen1998,Bovier04}. A formal definition of the multiscale measure  will be  given in Section \ref{sec:defs}. 

It is worth to stress  that a  multiscale measure does not describe a standard thermodynamic equilibrium but rather an out of equilibrium scenario. 
{Indeed, it can be viewed as the stationary measure of a dynamical system in which different degrees of freedom, say $(x_\ell)_{1\leq\ell\leq r}$, are coupled to thermal baths at different temperatures and evolve on widely separated timescales, $\tau_1 \gg \ldots \gg \tau_r$.
At a given time $t \sim \tau_\ell$, the dynamics can be effectively described on the variable $x_\ell$, where the faster degrees of freedom $(x_{\ell'})_{\ell' > \ell}$ have already equilibrated, and the slower ones $(x{_{\ell'}})_{\ell' < \ell}$ are effectively frozen.
The stationary measure of this effective dynamics is then obtained through the same recursive procedure outlined above.
A concrete example of such dynamics is discussed in Remark \ref{Dynamics}. For a thorough discussion of the physical aspects, we refer to \cite{Penney_1993,Allahverdyan_2000pre,Coolen_1993,Contucci_2019,Contucci_2021}, and for a mathematical analysis to \cite{alberici2024convergence}.
} In a broader sense the multiscale measure can be used to describe out of equilibrium systems in the limit of small entropy production \cite{1999cond.mat.11086C}: here the different temperatures are defined as the ratio between  correlation and response functions, generalizing the classical fluctuation-dissipation relation \cite{bouchaud1998out,Cugliandolo_2011,PMPF}. 

The multiscale measure spontaneously emerges in the \textit{Replica Symmetry Breaking} (RSB) solution of SK model. More specifically, the Parisi formula \cite{Parisi80} can be obtained from the free energy of a system with non-interacting spins subject to multiscale Gaussian random external fields with a special correlation structure \cite{guerra2003broken,TalagrandPArisi,Panchenko2013}. The Parisi formula implies that at low temperature RSB occurs \cite{AuChen13,Au2017,Jaga2015} and the physical behavior of the model becomes, to some extent, close to that of out of equilibrium systems \cite{MPV,Leuzzibook}. We also mention that mean-field models where the external field is governed by a multiscale measure have been recently considered in a series of works \cite{BarraGino,Mourrat-Panchenko,Bimourrat,dominguez2024statistical} leading to a formulation of the Parisi solution in terms of Hamilton-Jacobi equations.

Multiscale models also arise in high dimensional inference when multiple steps of inference procedures are concatenated. We mention as examples: expectation-maximization approaches, reconstruction tasks where a part of the hidden signal has been revealed \cite{cui2021large,cui2022large}, or matrix inference models \cite{Barbier_Camilli2024}.  Similar structures also emerge in the teacher-student analysis where the multibath property is related to a possible mismatch between inference and generating temperatures \cite{nishimori2001statistical,alemanno2023hopfield, theriault2024dense,theriault2024modelling}.  Multiscale multibath hierarchical procedures are also common in the context of machine learning, especially in the area of transfer learning and fine tuning techniques analysis \cite{gerace2022probing,li2023statistical,ingrosso2024statistical}.

In \cite{contucci-mingione} the multiscale SK model was shown to obey  a Parisi-like variational principle. In this work we study its solution and the related thermodynamic properties. In particular we 
describe how the multiscale structure affects the distribution of the order parameter.  Contrary to what happens in the SK model, here the system's behavior can be characterized through the overlap distribution at all the different time scales. In physical terms the overlap at the $\ell$-th scale is obtained by measuring the scalar product between the spin configurations of two system's replicas evolved up to the $\ell$-th equilibrating time from the same initialization. In this situation part of the interactions has evolved while the rest is still frozen. 

Based on the above physical picture the overlap is expected to decrease at increasing equilibrating time, since the greater the differences in interactions between the two replicas, the more distinctly their spin configurations will evolve. Furthermore, according to the conventional ultrametric picture of the states, at low temperature different timescales are expected to correspond to different ranges of overlap values, which are related to the corresponding different levels of the ultrametric tree \cite{MPV, Panchenko2013,Jagannath}. It is also reasonable to expect that, at sufficiently high temperature, all the $\ell$-averages of the overlap vanish and the system is completely annealed.
Interestingly, this model can include a partial annealing mechanism in which the overlap is  different from zero up to a certain timescale, beyond which it then becomes zero.

This heuristic picture is confirmed and specified by the results of this work. We start by showing how the $\ell$-th scale average of the overlap is related to the minimizer of the Parisi formula by the synchronization mechanism \cite{panchenko_multi-SK, franz1992ultrametricity,genovese2017overlap} and in Theorem~\ref{thm:overlap_moments} we compute it. Following that, in Theorem~\ref{thm:annealed} we characterize the high-temperature region of the model where in absence of external magnetic fields the free energy is yielded by an annealed computation, namely by considering all interactions as thermodynamic degrees of freedom at the same time scale of the spins. 
In Theorem~\ref{thm:minRSB} we give a sufficient low temperature type condition  to have  at least as many RSB levels as the time scales originally present in the model. Finally in Theorem \ref{thm:plateau} we characterize  the overlap distribution across the various scales providing a sufficient condition for the existence of gaps in its support.

The paper is organized as follows. In Section \ref{sec:defs} the SK multiscale measure is introduced and the Parisi like variational principle for the multiscale pressure is recalled. In Section \ref{sec:results} the main results are presented and proved in  Section \ref{sec:proof}. In  Section \ref{sec:conclusions} we present a synthetic picture of our finding together with   conclusions and perspectives.

\section{Definitions}\label{sec:defs}
Let us first define the pressure of the multiscale Sherrington-Kirkpatrick (SK) model \cite{contucci-mingione}. Consider an integer $r\geq 1$ and two sequences $\zeta=(\zeta_{\ell})_{\ell\leq r}$ and $\gamma=(\gamma_{\ell})_{\ell\leq r}$ such that
\begin{align}
    \label{eq:sequences1}
    &0=\zeta_{-1}<\zeta_0<\dots<\zeta_r=1\\
    \label{eq:sequences2}
    &0=\gamma_0<\gamma_1<\dots<\gamma_r<\infty\,,
\end{align}
and a set of $N$ spins $\sigma\equiv(\sigma_i)_{i\leq N}\in\{-1,1\}^N$ interacting via a random Hamiltonian
\begin{align}
    H_N(\sigma)=\sum_{\ell=1}^r H^{(\ell)}_{N}(\sigma)\,.
\end{align}
Analogously to the standard SK model, we take each of the Hamiltonians $H_N^{(\ell)}$ to be a $2^N$-dimensional Gaussian process indexed by the spin configurations, identified by the following covariance
\begin{align}\label{eq:hami}
    \mathbb{E}\, H^{(\ell)}_N(\sigma)H^{(\ell')}_N(\tau)=N\delta_{\ell\ell'}(\gamma_{\ell}^2-\gamma_{\ell-1}^2)q_N^2(\sigma,\tau)
\end{align}for any $\sigma,\tau\in\{-1,1\}^N$, where we have introduced the overlap
\begin{align}\label{def:overlap}
    q_N(\sigma,\tau)=\frac{1}{N}\sum_{i=1}^N\sigma_i\tau_i\,.
\end{align}
The same model can be equivalently represented in terms of independent and identically distributed standard Gaussians $g_{ij}^{(\ell)}$:
\begin{align}\label{def:hamiell}
    H_N^{(\ell)}(\sigma)=\sqrt{\gamma_{\ell}^2-\gamma_{\ell-1}^2}\sum_{i,j=1}^N\frac{g_{ij}^{(\ell)}}{\sqrt{N}}\sigma_i\sigma_j\,.
\end{align}

The thermodynamic pressure of the multiscale SK model is defined by recursive integration of the randomness in the different Hamiltonians $H_N^{(\ell)}$ at different  scales. More specifically, define the backwards recursion
\begin{align}\label{eq:recursion_main}
    Z_{\ell-1,N}^{\zeta_{\ell-1}}=\mathbb{E}_{\ell-1}Z_{\ell,N}^{\zeta_{\ell-1}}\,,\quad \mathbb{E}_{\ell-1}=\mathbb{E}_{(g^{(\ell)}_{ij})_{i,j\leq N}}
\end{align}for any $\ell=1,\dots,r$, with starting point
\begin{align}
    Z_{r,N}=\sum_{\sigma\in\{-1,1\}^N}e^{-H_N(\sigma)-\sum_{i=1}^Nh_i\sigma_i}\,,
\end{align}where $h_i\sim P_h$ are i.i.d.\ copies of a compactly supported quenched random variable.
We stress that even if \eqref{eq:hami} looks like a sum of independent Gaussians, the recursion \eqref{eq:recursion_main} induces non trivial dependencies, as later illustrated.

Given a realization of $h$, $Z_{r,N}$ depends on all the randomness, i.e. $(g_{ij}^{(\ell)})_{i,j\leq N}^{\ell\leq r}$, while $Z_{\ell,N}$ depends only on the randomness up to level $\ell$, that is $(g_{ij}^{(p)})_{i,j\leq N}^{p\leq\ell}$. We are now ready to define the main object under investigation:

\vspace{3pt}
\begin{definition}[Pressure per particle]
    The pressure per particle of the multiscale SK model is \begin{align}\label{eq:pressure_per_particle}
    p_N=\frac{1}{N}\EE_h\log Z_{0,N}\,.
\end{align}
\end{definition}

\vspace{3pt}
\begin{remark}
For  $r=1$ and $\zeta_0>0$ the model was already studied in \cite{TalASS,BAG}, while  for $\zeta_0\to0$ the quantity $p_N$ reduces to the quenched pressure of an SK model at inverse temperature $\beta=\gamma_r$.
\end{remark}

\vspace{3pt}
\begin{remark}
One can also define the $\ell$-level pressure  $P_{\ell,N}=\log Z_{\ell,N}$, then the recursion \eqref{eq:recursion_main} rewrites as 

\begin{equation}e^{\zeta_{\ell-1}\,P_{\ell-1,N}}=\mathbb{E}_{\ell-1}e^{\zeta_{\ell-1}\,P_{\ell,N}}.
\end{equation}

The above relation is common  in renormalization group approach to field theory (see \cite{Gallavotti}).
\end{remark}

\vspace{3pt}
\begin{remark}\label{RPC-repres}
The  pressure \eqref{eq:pressure_per_particle} can be also written as the  quenched pressure of  an auxiliary system with configuration space $\{1,-1\}^N\times \mathbb{N}^r$  

\begin{equation}\label{eq:pressureRPC}
p_N=\frac{1}{N}\,\E\,\log\sum_{\alpha\in\mathbb{N}^r}\nu_{\alpha}\sum_{\sigma\in\{-1,1\}^N} e^{\,H_N(\sigma,\alpha)}
\end{equation}
where $H_N(\sigma,\alpha)$ is a suitable centered Gaussian process $(\nu_{\alpha})_{\alpha\in\mathbb{N}^r}$ are random probabilistic weights  associated to a  Ruelle Probability Cascade \cite{Panchenko2013,contucci-mingione}. 
\end{remark}

\subsection{The multiscale measure}\label{sec:multiscale_measure}
The pressure \eqref{eq:pressure_per_particle} is the generating functional of the Hamiltonians $(H_N^{(\ell)})_{\ell\leq r}$ with respect to the \textit{multiscale measure}. In particular, the recursion \eqref{eq:recursion_main} implies that the average w.r.t. the multiscale  measure   is obtained by a sequence of Boltzmann-Gibbs averages of the different degrees of freedom, each of them performed at a proper temperature and effective  potential. In the current setting the degrees of freedom are the spins $\sigma\in\{-1,1\}^N$ and the collection of couplings $(g^{(\ell)})_{\ell\leq r}$ where $g^{(\ell)}=(g^{(\ell)}_{ij})_{i,j\leq N}$. The spins are the fastest variables (i.e.\ first to thermalize) and  if $\ell>\ell'$ $g^{(\ell)}$ is faster than $g^{(\ell')}$. Finally the external field $(h_i)_{i\leq N}$ is completely quenched, namely one has to average it at the end of the thermalization procedure. 

The above picture is formally defined as follows. The spins thermalize according to the standard Boltzmann-Gibbs measure given a realization of $g=(g^{(\ell)})_{\ell\leq r}$ and $h$, namely for any function $A(\sigma)$ we set 
\be\label{eq:BGdistribution}
\left\langle\,A\,\right\rangle_N=\sum_{\sigma\in\{-1,1\}^N}\dfrac{e^{-H_N(\sigma)+\sum_{i\leq N}h_i\sigma_i}}{Z_{r,N}} \,A(\sigma)\,.
\ee
The average of the remaining degrees of freedom is taken using a suitable collection of probability weights.  For any $\ell\leq r$ we set
\begin{align}\label{eq:tilts_ell}
    f_{\ell,N}:=\frac{Z_{\ell,N}^{\zeta_{\ell-1}}}{\EE_{\ell-1}Z_{\ell,N}^{\zeta_{\ell-1}}}\,,\quad 1\leq \ell\leq r\,.
\end{align} Notice that $(f_{\ell})_{\ell\leq r}$ are random probabilistic weights, in particular $f_\ell$ depends on the families $(g^{(\ell')})_{\ell'\leq\ell}$ and $h$. 

\vspace{3pt}
\begin{definition}\label{def:l-scales-average}
Given $\ell\in\{0,\ldots,r\}$  and a measurable function $A(\sigma,(g^{(\ell')})_{\ell<\ell'\leq r}))$ we define its
\textit{$\ell$-th scale average}  as
\begin{align}\label{eq:level-bracket0}
\langle\,A\,\rangle_N^{(\ell)}=\EE_{\ell}\EE_{\ell+1}\dots\EE_{r-1} \,f_{\ell+1,N}\dots f_{r,N}\langle\,A \,\rangle_N
\end{align}
with the convention  $\langle\,\cdot\,\rangle_N^{(r)}\equiv \langle\,\cdot\,\rangle_N$. 
\end{definition}
The above definition can be equivalently rewritten in terms of the recursion
\begin{align}
    \label{eq:recursion_bracket_original}
    \langle\cdot\rangle^{(\ell-1)}_N=\EE_{\ell-1}f_{\ell,N}\langle\cdot\rangle^{(\ell)}_N\,.
\end{align}

$\langle\,A\,\rangle_N^{(\ell)}$ is random through $(g^{(\ell')})_{\ell'\leq\ell}$ and $h$; from a probabilistic point of view it can be viewed as the conditional expectation given $(g^{(\ell')})_{\ell'\leq\ell}$ and $h$.\\

\begin{remark}\label{Dynamics}
{
The recursion \eqref{eq:recursion_bracket_original} can be interpreted as defining a multiscale measure induced by a generic Hamiltonian $H$ of a system with $r+1$ families of degrees of freedom. Let us denote these families by $x=(x_{\ell})_{\ell \leq r+1}$ and the average w.r.t.\ the randomness of $x_{\ell}$ by $\E_{\ell-1}$. The multiscale measure induced by $H$ is obtained by \eqref{eq:recursion_bracket_original} starting from
\be
\langle \cdot \rangle^{(r)}=\frac{1}{\E_{r} e^{-\beta H(x)}}\E_{r}e^{\beta H(x)}(\cdot)\,.
\ee

As mentioned in the introduction, there is a dynamical interpretation of the above construction. Let us assume to deal with continuous degrees of freedom  and consider the following system of Langevin equations:
\be
\tau_{\ell} \, \mathrm{d} x_{\ell} = -\partial_{\ell} H \, \mathrm{d}t + \sqrt{\frac{2\tau_{\ell}}{\beta_{\ell}}} \, \mathrm{d}W_{\ell}, \quad \ell = 1, \ldots, r+1
\ee
where $\tau_{\ell}, \beta_{\ell} > 0$, and $\{W_{\ell}\}_{\ell \leq r+1}$ are independent Wiener processes. Then, under the asymptotic regime $\frac{\tau_{\ell}}{\tau_{\ell+1}} \to \infty$, one can show (see \cite{contucci2021stationarization, alberici2024convergence}) that the stationary measure of the system converges to the recursive multiscale structure given by \eqref{eq:recursion_bracket_original}, with the identifications $\beta_{r+1}=\beta$ and $\zeta_{\ell} = \beta_{\ell+1} / \beta$.}
\end{remark}
\vspace{3pt}

The order parameter of the model turns out to be \cite{contucci-mingione} the distribution of overlap  \eqref{def:overlap}, which is a function of two spin configurations that in physical jargon are called \textit{replicas}. Clearly one has to specify from which distribution   spin configurations are sampled. In fact, in the multiscale setting there are $r+1$ different ways to sample them. 

\begin{definition}
Given $\ell\in\{0,\ldots,r\}$ we denote  by $\mu_N^{(\ell)}$ the measure in \eqref{eq:level-bracket0}. Let $A_{1,\dots,n}\equiv A(\sigma^{(1)},\dots,\sigma^{(n)})$ be a function of $n$ spin configurations that are sampled independently from $\mu^{(\ell)}_N$, the \textit{$\ell$-th level replicated average} is defined as 
\begin{align}\label{eq:replevel-bracket}
    \langle\,A_{1,\ldots,n}\,\rangle_N^{(\ell)}=\int \prod_{a=1}^n\mu^{(\ell)}_N(d\sigma^{(a)}) A(\sigma^{(1)},\dots,\sigma^{(n)})\,.
\end{align}

\end{definition}

Some remarks are in order. Recall that, as for the single replica average \eqref{eq:level-bracket0}, $\langle\,A\,\rangle_N^{(\ell)}$ is random through $(g^{(\ell')})_{\ell'\leq\ell}$ and $h$. Second, each replica $\sigma^{(a)}$ shares the same \enquote{outer} disorder $(g^{(\ell')})_{\ell'\leq\ell}$, but they all come with their own replicas of the couplings for the levels $(g^{(\ell')})_{\ell'>\ell}$. In this sense, the multiscale model presents a fundamental difference with the standard SK.

We are mostly interested in the case $n=2$ and $A_{12}\equiv q_N(\sigma^1,\sigma^2)$ where $q_N$ is  overlap \eqref{def:overlap}.  Hence, for any $\ell\in\{0,\ldots,r\}$ and  bounded  function $\psi$, we have the  expectation
\begin{align}\label{overlap average}
    \EE\langle \langle \psi(q_{1,2})\rangle^{(\ell)}_N\rangle^{(0)}=\EE_h\,\prod_{p=1}^{\ell} f_{p,N}\int \mu_N^{(\ell)}(d\sigma^1)\mu_N^{(\ell)}(d\sigma^2) \psi(q_{N}(\sigma^1,\sigma^2))\,.  
\end{align}
The average \eqref{overlap average} can be divided into three steps. Given $(g^{(\ell')})_{\ell'\leq \ell}$ we first take the average value of the overlap between two spin configuration sampled from $\mu_N^{(\ell)}$. The result is a random function of $(g^{(\ell')})_{\ell'\leq \ell}$ which is integrated using the tilted measure $\langle\, \cdot\,\rangle^{(0)}_N$. Finally we take the expectation $\E_h$  to average out the  quenched external field. 

The connection between the $\ell$-scale average \eqref{eq:level-bracket0} and the multiscale pressure is essentially due to the recursion  \eqref{eq:recursion_main} (see \cite{guerra2003broken}). 
As an example one can compute the various contributions to the  \textit{internal energy} coming from $(H^{(\ell)}_N)_{\ell\leq r}$. 
Setting $\beta_{\ell}=\sqrt{\gamma^2_\ell-\gamma^2_{\ell-1}}$  one obtains

{
\be\label{internell}
\frac{\partial p_N}{\partial\beta_{\ell}} \,=\,-\,\mathbb{E}\prod_{s=1}^r f_{s,N}\left\langle \,\dfrac{H_N^{(\ell)}}{\beta_{\ell}N}\, \right\rangle_N\,=\,-\,\mathbb{E}_h \left\langle \dfrac{H_N^{(\ell)}}{\beta_{\ell}N}\right\rangle^{(0)}_N.
\ee}
Since $H^{(\ell)}_N$ is a Gaussian one can apply integration by parts to rewrite the r.h.s. of \eqref{internell} as
\begin{equation}\label{eq:internderiv}
\frac{\partial p_N}{\partial\beta_{\ell}} =\beta_{\ell}\Big(1-\sum_{p=\ell}^r (\zeta_{p}-\zeta_{p-1})\, \E_h \Big \langle \,\langle q^2_{12}\rangle^{(p)}_N\Big\rangle^{(0)}_N \Big)
\end{equation}
where $\langle\,\cdot\,\rangle^{(p)}_N$ is the $p$-level replicated average defined in \eqref{eq:replevel-bracket}.
Moreover  the quantity $\frac{\partial p_N}{\partial\gamma_{\ell}}$ gives the contribution to the internal energy  due to  the $\ell$-average. From \eqref{eq:internderiv} one obtains
\begin{equation}\label{eq:gamma_der_pN}
\frac{\partial p_N}{\partial\gamma_{\ell}}\,=\, \begin{cases}-\gamma_{\ell}(\zeta_{\ell}-\zeta_{{\ell}-1})\,\E_h \Big \langle \,\langle q^2_{12}\rangle^{(\ell)}_N\Big\rangle^{(0)}_N\,,\quad \mathrm{if}\quad \ell=1,\ldots, r-1 \\ \gamma_r\Big(1-(1-\zeta_{r-1})\, \E_h \Big \langle \,\langle q^2_{12}\rangle^{(r)}_N\Big\rangle^{(0)}_N \Big)\,\quad \mathrm{if}\quad \ell=r.
\end{cases} 
\end{equation}

\subsection{The variational formula}
The limiting value as $N\to\infty$ of the pressure  \eqref{eq:pressure_per_particle} exists and  it can be represented as the  solution of an infinite dimensional variational problem \cite{contucci-mingione}. Let us introduce it first. Consider an integer $k\geq r$ and two sequences
\begin{align}\label{eq:solution_sequences}
    &0=\xi_{-1}\leq\xi_0\leq \dots\leq\xi_k=\xi_{k+1}=1\\
    \label{eq:solution_sequences2}
    &0=x_0\leq x_1\leq \dots\leq x_k\leq x_{k+1}=1\,,
\end{align}
where $\xi=(\xi_j)_{j=0,\dots,k+1}\supseteq\zeta=(\zeta_\ell)_{\ell=0,\dots,r}$. Given $(\gamma_\ell)_{\ell\leq r}$ in \eqref{eq:sequences2} we define 
\begin{align}\label{eq:tilde_gamma}
    \tilde\gamma_j=\gamma_\ell\quad\text{if}\quad j\in K_{\ell}=\{j:\zeta_{\ell-1}<\xi_j\leq\zeta_{\ell}\}\quad\text{for some}\,\,\ell\in\{0,\ldots,r\}
\end{align}
 and the one-body Hamiltonian
\begin{align}
    -\tilde H(\sigma)=\sigma\Big(\sqrt{2}\sum_{j=1}^{k+1}\eta_j\sqrt{\tilde\gamma^2_{j}x_j-\tilde\gamma^2_{j-1}x_{j-1}}+h\Big){=\sigma\Big(\sqrt{2}\sum_{\ell=0}^{r}\gamma_\ell\sum_{j\in K_\ell}\eta_j\sqrt{x_j-x_{j-1}}+h\Big)}\,,
\end{align}where $\sigma=\pm1$, $\eta_j\iid\mathcal{N}(0,1)$, and $h\sim P_h$. Starting from
\begin{align}
    \tilde Z_{k+1}:=\sum_{\sigma=\pm1} e^{-\tilde H(\sigma)} =2\cosh\big(\sqrt{2}\sum_{j=1}^{k+1}\eta_j\sqrt{\tilde\gamma^2_{j}x_j-\tilde\gamma^2_{j-1}x_{j-1}}+h\big)
\end{align} we define again the backwards recursion
\begin{align}\label{recurtilde}
    \tilde Z_{j-1}^{\xi_{j-1}}=\mathbb{E}_{j-1}\tilde Z_{j}^{\xi_{j-1}}\,,\quad \mathbb{E}_{j-1}=\mathbb{E}_{\eta_j}\,,
\end{align}for $j=1,\dots,k+1$. This finally allows to introduce the \emph{Parisi functional} for the multiscale model as
\begin{align}
    \label{eq:Parisi_functional}
    \mathcal{P}(x,\xi)=\log \tilde Z_{0}-\frac{1}{2}\sum_{j=0}^{k}\xi_{j}\big((\tilde\gamma_{j+1}x_{j+1})^2-(\tilde\gamma_jx_j)^2\big){
    =\log \tilde Z_{0}-\frac{1}{2}\sum_{\ell=0}^r\gamma_\ell^2\sum_{j\in K_\ell} \xi_{j}(x_{j+1}^2-x_j)
    }\,.
\end{align}

Out of convenience we introduce the following sets of allowed sequences
\begin{align}
\mathcal{M}=\bigcup_{k\geq r}\mathcal{M}_k\,,\quad\mathcal{M}_k:=\{
    (x,\xi)\text{ verifying \eqref{eq:solution_sequences}-\eqref{eq:solution_sequences2}},\,\xi\supseteq\zeta,\, \text{card}(x)=k+2
    \}.
\end{align} We can now state the main result of \cite{contucci-mingione}, as it plays a central role in our analysis.
\vspace{3pt}

\begin{theorem}\label{them:pressure}The thermodynamic limit of the quenched pressure density of the Multiscale SK (MSK) model \eqref{eq:pressure_per_particle} exists and is given by the infinite dimensional variational principle
\begin{align}\label{eq:var_principle}
\lim_{N\to\infty}p_N=\inf_{(x,\xi)\in\mathcal{M}}\mathcal{P}(x,\xi)\,.
\end{align}
\end{theorem}
\vspace{3pt}

The proof of this theorem requires a non-trivial extension of the Guerra replica symmetry breaking upper bound \cite{guerra2003broken} and the use of the synchronization property \cite{panchenko_multi-SK} for Ruelle Probability Cascades (RPCs) for the lower bound. We refer the reader to \cite{contucci-mingione} for the details. Here we aim at assessing some fundamental properties of the optimization over $\mathcal{M}$. The problem of uniqueness of the solution is not addressed here. However, since  Theorem \ref{them:pressure} has many analogies with the Parisi formula for the SK model where uniqueness holds \cite{AuChen13, Jatoga}, it is reasonable to expect the same here.

\section{Results}\label{sec:results}

\subsection{The order parameter and synchronization}\label{sec:synchro}
In this section we analyze the continuity properties of $\mathcal{P}$ seen as a function over $\mathcal{M}$. We shall see that, contrary to the plain SK model, in this multiscale version $\mathcal{P}$ is not only a function of the probability distribution associated to the sequences $(x,\xi)$, but it preserves some memory of the scales originally present, identified by the sequences $(\gamma,\zeta)$, in a sense that will be rigorously specified.

Let us denote by $\mathrm{Pr}$ the space of probability measures supported on $[0,1]$. Define $\mathrm{Pr}_k$, for $k$ integer, the set of probability measures supported on $k+1$ points on $[0,1)$. Any $\mu\in\mathrm{Pr}_{k}$ can be identified with a pair of strictly increasing sequences $(y,m)$ such that
\begin{align}
0=m_{-1}<m_0<\dots<m_{k}=1\\
0=y_0<y_1<\dots...<y_{k}<1\,.
\end{align}The mapping to such distribution is realized by the following
\be
\mu(y_i)=
m_{i}-m_{i-1}\,,\,\quad i=0,\ldots,k\,.
\ee

\begin{definition}[Quantile]
Let $\rho$ be a probability measure supported on $[0,a]$ for some $a>0$. The quantile function associated to $\rho$ is 
\be\label{def quantile}
\rho^{-1}(p):=\inf\{s\in[0,a]:\rho([0,s])\geq p \}\,,\quad p\in[0,1].
\ee
\end{definition}

\vspace{3pt}

Consider a $k'$ integer and $\mu\equiv (y,m)\in\mathrm{Pr}_{k'}$. Set $k:=|m\cup\zeta|\geq k'$. Note that the sequence $m$ can already contain some of the $\zeta$'s, or even all of them. We define a pair $(x^{\mu},\xi^{\mu})\in\mathcal{M}_k$ associated to $\mu$ by 
\begin{equation}\label{map}
\begin{cases}
\begin{aligned}
\xi^{\mu}\quad  &\,\text{take $m\cup\zeta$  increasingly ordered and set $\xi^{\mu}_{k+1}=1$}\\
x^{\mu}_j\,=&\mu^{-1}({\xi_{j}^{\mu}})\,, \quad \mathrm{for}\,\,j=0\ldots,k\,\,\text{and set $x^{\mu}_{k+1}=1$}\,.
\end{aligned}
\end{cases}
\end{equation}
\begin{figure}[h!!!]
    \centering
    \begin{tikzpicture}
\begin{axis}[
    xlabel={x},
    ylabel={$F_\mu$},
    xmin=0, xmax=1,
    ymin=0, ymax=1,
    scale=1
]
\addplot[color=blue,very thick,domain=0:0.3, 
    samples=100] {.2} node[above,pos=.5] {};
\addplot[style=dashed,domain=0:1, 
    samples=100] {.2} node[above,pos=.05] {$m_0$};

\addplot[color=blue,very thick,domain=0.3:0.5, 
    samples=100] {.6} node[above,pos=.2] {};
\addplot[style=dashed,domain=0:1, 
    samples=100] {.6} node[above,pos=.05] {$m_1$};

\addplot[color=blue,very thick,domain=0.5:0.8, 
    samples=100] {.7} node[above,pos=.2] {};
\addplot[domain=0:1,style=dashed, 
    samples=100] {.7} node[above,pos=.1] {$m_2={\color{red}\zeta_2}$};
 
\addplot[color=blue,ultra thick,domain=0.8:1, 
    samples=100] {1} node[below,pos=.2] {};

\addplot[domain=0:1,style=dashed, 
    samples=100,color=red] {.3} node[above,pos=.05] {$\zeta_0$};

\addplot[domain=0:1,style=dashed, 
    samples=100,color=red] {.5} node[above,pos=.05] {$\zeta_1$};

\addplot[domain=0:1,style=dashed, 
    samples=100,color=red] {.85} node[above,pos=.05] {$\zeta_3$};

\addplot[domain=0:1,style=dashed, 
    samples=100,color=red] {.7} node[above,pos=.05] {};

\node[above] at (0.08,0) {\small $x_0=0$};
\filldraw[blue] (0,0.2) circle (3pt);

\filldraw[red] (0.3,0.3) circle (3pt);

\filldraw[red] (0.3,0.5) circle (3pt);

\filldraw[blue] (0.3,0.6) circle (3pt);

\filldraw[blue] (0.5,0.7) circle (3pt);

\filldraw[red] (0.8,0.85) circle (3pt);

\filldraw[blue] (0.8,1) circle (3pt);

\end{axis}
\end{tikzpicture}
\caption{Example of construction of the sequence $(x^\mu,\xi^\mu)$ associated to a probabilty measure $\mu$ with CDF $F_{\mu}$ plotted in blue. In this example we have  $k=7$, $m_0<\zeta_0,\zeta_1<m_1$, $\zeta_2=m_2$, and $m_2<\zeta_3<1$. It results in the sequence $x^\mu=\{0,y_1,y_1,y_1,y_2,y_3,y_3\}$, that corresponds to the horizontal coordinates of the colored circles in the figure. Blue circles are due to the elements already present in the couple of sequences $(m,y)$, whereas red circles come from the newly introduced elements which form $\xi$. As apparent from this figure, repetitions occur all the times we are adding a $\zeta$ to the sequence $m$ which was not already contained in the latter.
}\label{fig:ximu_xmu}
\end{figure}
The fact that $x^\mu$ is defined through the quantile function $\mu^{-1}$ implies that the elements of the sequence $x^\mu$ are those of the sequence $y$ with possibly some repetitions. There are no repetitions in the $x^\mu$'s iff $m\supseteq\zeta$. On the other hand, if $m$ does not contain, say, $\zeta_\ell=\xi^\mu_{j_\ell}$ for a given $\ell=0,\dots,r$, then by definition of the quantile {$\mu^{-1}(\xi^\mu_{j_\ell})=x^\mu_{j_\ell}=x^\mu_{j_\ell+1}$}. In other words, if there are repetitions, they occur in correspondence to those $\zeta_\ell$'s missing from the sequence $m$ (see Fig.\ref{fig:ximu_xmu} for an example).

This reasoning is necessary precisely because the Parisi functional $\mathcal{P}$ does not depend only on the probability distribution identified by the pair $(x,\xi)$, but also on which of the $x$'s repeat. As we shall clarify later, if $x_{j+1}$ and $x_{j}$ collapse on one another, but $\xi_{j}\not\in\zeta$, then the Parisi functional becomes independent on $\xi_j$. It is by all means as if we had reduced the sequences by one element. This however is no longer true if $\xi_{j}\in\zeta$. Following definition \eqref{eq:Parisi_functional} it is indeed not difficult to see that all the $\zeta$'s appear anyway in the functional, as they should.

That being said, \eqref{map} defines a map $\bigcup_{k\geq1}\mathrm{Pr}_k\ni\mu\to (x^{\mu},\xi^{\mu}) \in\mathcal{M}$. We set 
\be
\bar{\mathcal{P}}(\mu):=\mathcal{P}(x^{\mu},\xi^{\mu})\,\,,\;\quad\quad\mu\in\bigcup_{k\geq1}\mathrm{Pr}_k\,.
\ee
We can now state a proposition concerning the Lipschitz continuity of $\bar{\mathcal{P}}$, and relating the latter to the Parisi functional. Its proof is deferred to Section \ref{sec:Lipschitz_proof}.

\vspace{3pt}
\begin{proposition}\label{continuosex}
Given  $\mu_1\in \mathrm{Pr}_{k_1}$ and $\mu_2\in \mathrm{Pr}_{k_2}$  one has
\begin{equation}\label{Lipschtiz}
\mid \bar{\mathcal{P}}(\mu_1)-\bar{\mathcal{P}}(\mu_2)\mid \,\leq\, L W_1(\mu_1,\mu_2)
\end{equation}
for some constant $L$ independent from $\mu_1$ and $\mu_2$ and
\begin{align}
    W_1(\mu_1,\mu_2):=\int^1_0 \mid \mu_1^{-1}(p)-\mu_2^{-1}(p)\mid dp\,.
\end{align} Therefore one can extend continuously $\bar{\mathcal{P}}$ to all $\mu\in\mathrm{Pr}$ and
\be\label{infequality}
\inf_{ (x,\xi)\in\mathcal{M}}\mathcal{P}(x,\xi)=\inf_{\mu\in\mathrm{Pr}}\bar{\mathcal{P}}(\mu). 
\ee

\end{proposition}

\begin{remark}\label{rem:redundacy_gammatilde_ell0}
Notice also that 
by definition \eqref{eq:tilde_gamma} one has $\tilde{\gamma}_j=0$ for all $j$ such that $0<\xi_j\leq\zeta_0$, hence these $\xi_j$ don't play any role in $\mathcal{P}(x,\xi)$ and then, without loss of generality, one can assume that $\xi_0=\zeta_0$. In other words one say that 
\be
\inf_{\mu\in\mathrm{Pr}}\bar{\mathcal{P}}(\mu)\,=\inf_{\mu\in\mathrm{Pr}^{(0)}}\bar{\mathcal{P}}(\mu)\quad\text{where} 
\quad
\mathrm{Pr}^{(0)}=\{\mu\in\mathrm{Pr}:\lim_{x\to 0^+}\mu([0,x])\geq\zeta_0\}\,.
\ee
\end{remark}

The sequences $(x^{\mu},\xi^{\mu})$ and the associated $\tilde{\gamma}$ defined in \eqref{eq:tilde_gamma} can be related to a pair of \textit{synchronized} random variables using a  construction  similar to  \cite{Mourrat-Panchenko,MourMonotone}.
We start identifying  the sequences $(\gamma,\zeta)$ in \eqref{eq:sequences1} and \eqref{eq:sequences2} with a r.v. $\Gamma$ setting 
{
\begin{align*}\label{eq:mu_Gamma_Def}
    \mu_\Gamma(\gamma_\ell):=\mathbb{P}(\Gamma=\gamma_{\ell})=\zeta_{\ell}-\zeta_{\ell-1}\,,\;\forall\; 0\leq\ell\leq r .
\end{align*}
}

\vspace{3pt}
\begin{definition}
Let $U\sim\mathrm{Unif}[0,1]$. Given a random variable $Y$ supported on $[0,1)$ with law $\mu$, we denote by  $\rho_{\mu}$ the probability measure $[0,1]\times[0,\gamma_r]$  identified by
\be\label{def:varC}
\rho_{\mu}=\mathrm{Law}(C)\quad \text{where}\quad 
C\equiv\left(\mu^{-1}(U) ,\mu^{-1}_{\Gamma}(U) \right)\ee
or equivalently in terms of CDF 
\be
F_C(c_1,c_2):=\min\Big(\mu([0,c_1]),\mu_\Gamma([0,c_2])\Big )\,,\quad(c_1,c_2)\in[0,1]\times[0,\gamma_r].
\ee
For any $\ell\leq r$ we denote by $\rho_{\mu,\ell}$ the conditional probability  
\be\label{eq:conditional_l_level}
\rho_{\mu,\ell}(A):=\int_A \rho_{\mu}(dx\mid\Gamma=\gamma_\ell)=\frac{1}{\zeta_\ell-\zeta_{\ell-1}}\int_A \rho_{\mu}(dx\,,\,\Gamma=\gamma_\ell).
\ee
\end{definition}
From the above definition, follows that the marginals of $C$ are $\Gamma$ and $Y$, but its components are strongly correlated random variables: they are deterministic functions of the same uniform random variable $U$. This construction  is known as \textit {synchronization} \cite{panchenko_multi-SK} or monotone coupling \cite{MourMonotone}. One can check that if $\mu\in\mathrm{Pr}_{k'}$ is a discrete probability measure, then the map \eqref{map} and the associate $\tilde{\gamma}$ coincide with the following definition
\begin{equation}\label{relation11}
\tilde{\gamma_j}=\mu^{-1}_{\Gamma}(\xi^{\mu}_j)\quad\text{and}\quad 
x^{\mu}_j=\mu^{-1}(\xi^{\mu}_j)\,.
\end{equation}
Furthermore, averages w.r.t. $\rho_{\mu,\ell}$ rewrite in terms of $\xi$ and $x$ as follows
\begin{align}
    \int_0^1 f(x)\,
    \rho_\mu(dx\mid\Gamma=\gamma_\ell)&=\frac{1}{\zeta_\ell-\zeta_{\ell-1}}\int_0^1 f(x)\,
    \rho_\mu(dx,\Gamma=\gamma_\ell)\nonumber\\
    &=\frac{1}{\zeta_\ell-\zeta_{\ell-1}}{\EE }f(\mu^{-1}(U))\mathbbm{1}(\mu_\Gamma^{-1}(U)=\gamma_\ell)\,.
\end{align}Notice that $\mathbbm{1}(\mu_\Gamma^{-1}(U)=\gamma_\ell)$ is non-zero only if $U\in(\zeta_{\ell-1},\zeta_{\ell}]$. Said interval can be decomposed as the union of $(\xi_{j-1},\xi_j]$ for $j\in K_\ell$, over which $\mu^{-1}(U)$ is constant and equal to $x_j$. Hence the above turns into
\begin{align}\label{eq:explicit_conditional}
    \int_0^1 f(x)\,
    \rho_\mu(dx\mid\Gamma=\gamma_\ell)&=\frac{1}{\zeta_\ell-\zeta_{\ell-1}}\sum_{j\in K_{\ell}}(\xi_j-\xi_{j-1})f(x_j)\,.
\end{align}

\subsection{Overlap moments}

By Proposition \ref{continuosex} we know that a solution of the variational problem \eqref{eq:var_principle} is a probability measure on $[0,1]$. Here we  link the above solution  to the distribution of  the overlap \eqref{def:overlap} w.r.t. the large $N$ limit of the multiscale measure \eqref{eq:replevel-bracket}.  
Let us recall that in the SK model the following holds (see \cite[Theorem 3.7]{Panchenko2013}, and \cite{diff_parisi}): the optimizer of the Parisi formula can be used to obtain the second moment of the overlap in the large $N$ limit of the quenched measure, which is in turn connected to the internal energy of the system. Our next result contains an analogous statement for the multiscale SK model. We stress that in the multiscale setting the picture is more involved since, as observed (see Section \ref{sec:multiscale_measure}), there are several possible measures for the overlap. One can consider for instance all the $\ell$-th scale averages $\E_h \Big \langle \,\langle q^2_N\rangle^{(\ell)}_N\Big\rangle^{(0)}_N$, $\ell=1,\ldots,r$.

\vspace{3pt}

\begin{theorem}[Internal energy and overlap moments]\label{thm:overlap_moments}
For any $\ell=1,\dots,r$ 
\begin{align}\label{eq:internalenergy}
\lim_{N\to\infty}
\E_h \Big \langle \,\langle q^2_N\rangle^{(\ell)}_N\Big\rangle^{(0)}_N 
=\int x^2\,\rho_{\mu^*}(dx\mid \Gamma=\gamma_\ell)\,.
\end{align}
where $\mu^*$ is some Parisi measure solving \eqref{infequality}.
\end{theorem}
There is a natural ordering in the moments of the overlap induced by the recursion \eqref{eq:recursion_bracket_original} and the fact that $f_{\ell,N}$ are probability tilts (namely $\EE_{\ell-1} f_{\ell,N}=1$). More precisely, we have
\begin{align}
    &\E_h \Big \langle \,\langle q^2_N\rangle^{(\ell)}_N\Big\rangle^{(0)}_N =\frac{1}{N}+{\frac{N-1}{N}}\E_h \Big \langle(\langle \sigma_1\sigma_2 \rangle^{(\ell)}_N)^2\Big\rangle^{(0)}_N=\frac{1}{N}+{\frac{N-1}{N}}\E_h \Big \langle(\EE_\ell f_{\ell+1}\langle \sigma_1\sigma_2 \rangle^{(\ell+1)}_N)^2\Big\rangle^{(0)}_N\nonumber\\
    &\qquad \leq \frac{1}{N}+{\frac{N-1}{N}}\E_h \Big \langle(\langle \sigma_1\sigma_2 \rangle^{(\ell+1)}_N)^2\Big\rangle^{(0)}_N=\E_h \Big \langle \,\langle q^2_N\rangle^{(\ell+1)}_N\Big\rangle^{(0)}_N
\end{align}
where in the last step we have reabsorbed $\EE_{\ell}f_{\ell+1}$ in $\langle\cdot \rangle^{(0)}_N$ and used Jensen's inequality. Hence one should expect that the r.h.s.\ of \eqref{eq:internalenergy} respects the same ordering. This is indeed the case thanks to the fact that $\rho_{\mu}$ is generated through a monotone coupling. Hence if the r.v.\ $\Gamma$ takes higher values, so does the corresponding moment of the overlap.

Note that, whereas the points of the supports of $\rho_{\mu^*}(x\mid\Gamma=\gamma_\ell)$ are disjoint for different $\ell$'s because of the monotone coupling, this does not imply that the finite size multiscale measures $\EE_h\langle\langle\cdot\rangle^{(\ell)}_N\rangle^{(0)}_N$ also have disjoint and ordered supports. That being said, if we stick to the interpretation that $\rho_{\mu^*}(x\mid\Gamma=\gamma_\ell)$ is the asymptotic conditional average for the $\ell$-th scale, Theorem \ref{thm:overlap_moments} is telling us that the overlaps become synchronized with $\Gamma$ and the supports of the asymptotic $\ell$-th scale measures separate.

\subsection{The annealed regime}
Our next result concerns the annealed solution of the model. Indeed, we prove that in absence of external magnetic fields an annealed solution holds iff a high-temperature condition is fulfilled:
\vspace{3pt}

\begin{theorem}[Annealed region]\label{thm:annealed}
The annealed solution is exact whenever $h=0$ almost surely and $\gamma^2_r\leq \frac{1}{2}$, namely
    \begin{align}\label{cond:ann}
        \lim_{N\to\infty} p_N\,=\,\log 2 +\dfrac{\gamma^2_r}{2}\,\text{ if and only if }\,  \gamma^2_r\leq \frac{1}{2}\text{ and }h=0 \text{ a.s.}
    \end{align}
\end{theorem}

As we shall see, a consequence of Theorem \ref{thm:overlap_moments} and \ref{thm:annealed} is that in the annealed region 
\be\label{zerooveralp}
\lim_{N\to\infty}
\E_h \Big \langle \,\langle q^2_N\rangle^{(\ell)}_N\Big\rangle^{(0)}_N =0,\,
\ee
for all $\ell\in\{1,\ldots,r\}$.
The annealed solution is an obvious upper bound for the pressure per particle thanks to Jensen's inequality.  The proof of Theorem \ref{thm:annealed} is reported in Section \ref{sec:proofthm_annealed}. However, proving it is exact requires {a direct use of the high-temperature result for the standard SK model an absence of external magnetic field
\cite{Tala_vol1}}.

\subsection{ Strong coupling regime}
The annealed condition \eqref{cond:ann}, namely when all the couplings $(\gamma_{\ell})_{\ell\leq r }$ are weak, implies that all the $\ell$-averages of the overlap are zero \eqref{zerooveralp}, and thus all the mass of the Parisi measure is concentrated at $0$.  
Conversely, if the interactions are strong enough, this is no longer true. {In particular we provide a sufficient strong coupling condition such that the distribution $\mu^*$, where the infimum of $\mathcal{P}$ is attained, has at least $r$ distinguished atoms.} Specifically, the (finite or infinite) sequence $x^*$, associated to $\mu^*$, must have at least $r$ distinguished values. In this sense, the model must have at least $r$ levels of replica symmetry breaking.

\vspace{3pt}

\begin{theorem}\label{thm:minRSB}
Assume that 
{
\begin{align}\label{eq:low-T-condition}
        1-2\gamma_1^2\zeta_1^2<0
\end{align}or $\EE_hh^2>0$.} Let $\mu^*\in\text{Pr}$ be a Parisi measure solving \eqref{infequality}. If $\mu^*$ is supported on a finite number of points then $\zeta\subseteq\xi^*$ and 
\be\label{positiveoverlap}
\lim_{N\to\infty}
\E_h \Big \langle \,\langle q^2_N\rangle^{(\ell)}_N\Big\rangle^{(0)}_N >0\quad{\forall 1\leq\ell\leq r\,.}
\ee
\end{theorem}

{
Condition \eqref{eq:low-T-condition} is actually the strongest among the set of conditions $\big(\gamma_\ell^2>\frac{1}{2\zeta_\ell^2}\big)_{\ell\leq r}$ involved in Lemma \ref{lem:partial_annealing--} later proved in Section \ref{sec:proofthm_minRSB}. Informally, what Lemma \ref{lem:partial_annealing--} is saying is that if $\gamma_\ell^2>\frac{1}{2\zeta_\ell^2}$ (in absence of an external magnetic field) and the measure $\mu^*$ is finitely supported, then it must have an atom detached from the origin, and it is the one in correspondence of the element $\zeta_\ell$ fo the sequence $\zeta$. This fact will be also crucial to assert that $\zeta_\ell\in\xi^*$.

By comparing \eqref{zerooveralp} and \eqref{positiveoverlap}, one may wonder if there exists a region of the parameters where only a subset of the $\ell$-averages are non-zero, indicating a mechanism of \textit{partial annealing}. Theorem \ref{thm:minRSB}, and in particular Lemma \ref{lem:partial_annealing--} are indeed leaving this possibility open. Strictly speaking, if $\gamma_\ell^2>\frac{1}{2\zeta_\ell^2}$ the only thing we can conclude is that $\lim_{N\to\infty}
\E_h \Big \langle \,\langle q^2_N\rangle^{(\ell)}_N\Big\rangle^{(0)}_N >0$, but we were not able to prove also the converse statement, as is instead possible for the total annealing dealt with in Theorem \ref{thm:annealed}. We finally notice that Theorem \ref{thm:minRSB} does not rule out the possibility of full RSB, namely that $\mu^*$ is supported on non-discrete set.}

\vspace{3pt}

\subsection{Lower bound on gaps sizes}
The final result is a sufficient condition for which, at fixed $\ell$, the CDF associated to $\mu^*$ exhibits a plateau at height $\zeta_{\ell}$. This can be identified as a discontinuity in the associated quantile function ${\mu}^{*-1}$. Recalling that the quantile is a non-decreasing and left-continuous function, we define for any $\nu \in Pr$ the quantity
\be
\Delta_{\ell}(\nu):={\nu}^{-1}(\zeta_{\ell}^+)-{\nu}^{-1}(\zeta_{\ell})
\ee 
where ${\nu}^{-1}(\zeta^+_{\ell}):= \lim_{p\to\zeta_{\ell}^+}\nu^{-1}(p)$. If $\Delta_{\ell}(\nu)\geq C$ for some positive $C>0$ then the CDF of $\nu$ has a plateau at height $\zeta_{\ell}$  of length at least $C$ and hence a gap in its support. This possibility is stated in the following theorem, where $\mu^*$ is assumed to be the $W_1$-limit of a sequence of $k$-stationary pairs, {introduced in Definition \ref{def:stationary_pair} below}.

\vskip 0.3cm
\begin{theorem}\label{thm:plateau}
Given  $\ell=1,\ldots,r-1$, if $\zeta_{\ell}\,\gamma^2_{\ell+1}\,< 1/2$ then 
 \be\label{iceplatthem}
\Delta_{\ell}(\mu^*)\geq 2\,\left(\gamma^2_{\ell+1}-\gamma^2_{\ell}\right) \mu^{*-1}(\zeta_{\ell})\, \Big(\int^1_ {\mu^{*-1}(\zeta^+_{\ell})}\mu^*([0,x])\,dx \Big)^2 \,.
\ee
\end{theorem}

Theorem \ref{thm:plateau} can be interpreted as a result for the support of the conditional overlap distribution $\rho_{\mu^*,\ell}$. In fact it implies that there is a gap  between the support of $\rho_{\mu^*,\ell}$ and that of $\rho_{\mu^*,\ell-1}$. Note that the result is informative only outside the annealed region, provided that $\mu^{*-1}(\zeta_{\ell})>0$ and {$\mu^{*-1}(\zeta^+_{\ell})<1$.} Notice also that if $\zeta_{\ell}\,\gamma^2_{\ell+1}\,< 1/2$ for some $\ell$ then $\zeta_{\ell'}\,\gamma^2_{\ell'+1}\,< 1/2$  for all $\ell'<\ell$. 

\section{Proofs}\label{sec:proof}

\subsection{Preliminaries}
Let us start with some basic properties of the Parisi functional \eqref{eq:Parisi_functional}.

\vskip 0.2cm

\begin{definition}\label{def:stationary_pair} The pair of sequences $(\bar{x},\bar{\xi})\in\mathcal{M}_k$ is called $k$-stationary pair if
\be
\inf_{(x,\xi)\in\mathcal{M}_k}
\mathcal{P} (x,\xi) = \mathcal{P}(\bar{x},\bar{\xi}).
\ee
\end{definition}

In analogy with \eqref{eq:tilts_ell} we introduce the random probability weights
\begin{align}\label{eq:tilts}
    f_j:=\frac{\tilde Z_j^{\xi_{j-1}}}{\EE_{j-1}\tilde Z_j^{\xi_{j-1}}}\,,\quad 1\leq j\leq k+1\,,
\end{align}
where the $\tilde Z_j$'s are defined in \eqref{recurtilde}.
Recall that $f_j$ depends on the remaining randomness in $\eta_1,\dots,\eta_{j}$. We denote the $j$-th scale average as 
\begin{align}\label{eq:level-bracket}
    \langle A\rangle^{(j)}=\EE_{j}\EE_{j+1}\dots\EE_{k} \,f_{j +1}\dots f_{k+1}\langle A\rangle\,,\quad 0\leq j\leq k+1
\end{align}where 
\begin{align}
    \langle A \rangle=\frac{1}{\tilde{Z}_{k+1}}\sum_{\sigma=\pm 1}A(\sigma,(g^{(\ell')})_{\ell<\ell'\leq r})) e^{-\tilde H(\sigma)}\,.
\end{align}
Hence $\langle\cdot\rangle^{(j)}$ is still random through $\eta_1,\dots,\eta_{j}$, but $\eta_{j+1},\dots,\eta_{k+1}$ have been averaged out. By definition \eqref{eq:level-bracket} one has 
\begin{align}\label{eq:bracket-recursion}
    \langle\cdot\rangle^{(j-1)}=\EE_{j-1} f_{j}\langle\cdot\rangle^{(j)}\,.
\end{align}
Moreover the recursion \eqref{recurtilde} implies the following
\vspace{3pt}

\begin{proposition}\label{prop:guerra}
Denote by $\partial$ a generic derivative w.r.t. a variable in $\tilde Z_{k+1}$, be it $\tilde\gamma$ or $x$, and by $\EE$ the expectation over all the disorder $\eta_1,\dots,\eta_{k+1}$. Then
    \begin{align}
        \partial\log \tilde Z_0 =\EE\prod_{j=1}^{k+1} f_j\frac{1}{\tilde Z_{k+1}}\partial \tilde Z_{k+1}\,.
    \end{align}
    Furthermore, consider $p>j$:
    \begin{align}\label{eq:f-derivative1}
        &\partial_{\eta_j} f_j=\sqrt{2(\tilde\gamma_{j}^2x_j-\tilde\gamma_{j-1}^2x_{j-1})}f_j\langle\sigma \rangle^{(j)}\xi_{j-1}\\
        \label{eq:f-derivative2}
        &\partial_{\eta_j} f_p=\sqrt{2(\tilde\gamma_{j}^2x_j-\tilde\gamma_{j-1}^2x_{j-1})}f_p\xi_{p-1}\big(\langle\sigma\rangle^{(p)}-\langle\sigma\rangle^{(p-1)}\big)\,.
    \end{align}
\end{proposition}
\noindent For its proof we refer to \cite{guerra2003broken}. Here we prove a specialization of the above for the Parisi functional.

\vspace{3pt}

\begin{lemma}\label{lem:derivative_Parisi_x}
The gradient components of the Parisi potential w.r.t. $x$ read:
    \begin{align}\label{eq:Parisi-derivative}
        \partial_{x_j}\mathcal{P}(x,\xi)=\tilde\gamma_j^2(\xi_j-\xi_{j-1})\big[x_j-\EE\prod_{p=1}^j f_p\cdot(\langle \sigma\rangle^{(j)})^2\big]
    \end{align}where
    \begin{align}
        \langle \sigma\rangle^{(j)}=\EE_{j}\dots\EE_{k} f_{j+1}\dots f_{k+1}\tanh z(\eta)\,,\quad z(\eta):=\sum_{j=1}^{k+1}\eta_j\sqrt{2\big(\tilde\gamma_{j}^2 x_{j}-\tilde\gamma_{j-1}^2x_{j-1}\big)}+h\,.
    \end{align}
\end{lemma}
\begin{proof}
    Let us start with the second term in \eqref{eq:Parisi_functional}:
    \begin{align}
        \partial_{x_j}&\frac{1}{2}\sum_{0\leq p\leq k}\xi_p\big((\tilde\gamma_{p+1}x_{p+1})^2-(\tilde\gamma_p x_p)^2\big)=(\xi_{j-1}-\xi_j)\tilde\gamma_{j}^2x_j.
    \end{align}
    Using Proposition \ref{prop:guerra}, we can derive the functional $\log \tilde{Z}_0$:
    \begin{align}
        \partial_{x_j}\log \tilde Z_0&=\EE\prod_{p=1}^{k+1} f_p\tanh (z(\eta))\Big[\frac{\tilde\gamma^2_{j}\eta_{j}}{\sqrt{2\big(\tilde\gamma_{j}^2 x_{j}-\tilde\gamma_{j-1}^2x_{j-1}\big)}}-\frac{\tilde\gamma^2_{j}\eta_{j+1}}{\sqrt{2\big(\tilde\gamma_{j+1}^2 x_{j+1}-\tilde\gamma_j^2x_j\big)}}\Big]\nonumber\\
        &=\tilde\gamma_j^2\EE\prod_{p=1}^{k+1} f_p\big[1-\tanh^2 (z(\eta))\big]-\tilde\gamma_j^2\EE\prod_{p=1}^{k+1} f_p\big[1-\tanh^2 (z(\eta))\big]\nonumber\\
        &\qquad+\frac{\tilde\gamma_j^2}{\sqrt{2\big(\tilde\gamma_{j}^2 x_{j}-\tilde\gamma_{j-1}^2x_{j-1}\big)}} \sum_{p\geq j}^{k+1}\EE f_1\dots \partial_{\eta_j}f_p\dots f_{k+1}\cdot \tanh (z(\eta))\nonumber\\
        &\qquad\qquad-\frac{\tilde\gamma_j^2}{\sqrt{2\big(\tilde\gamma_{j+1}^2 x_{j+1}-\tilde\gamma_{j}^2x_{j}\big)}}\sum_{p\geq j+1}^{k+1}\EE f_1\dots \partial_{\eta_{j+1}}f_p\dots f_{k+1}\cdot \tanh (z(\eta))\Big]\,.
    \end{align}
    Using the formulae for derivatives \eqref{eq:f-derivative1}-\eqref{eq:f-derivative2} we get
    \begin{align}\label{xderivative}
        \partial_{x_j}\log \tilde Z_0&=\tilde\gamma_j^2\EE f_1\dots f_{k+1}\langle\sigma\rangle\sum_{p\geq j}^{k+1}\langle\sigma\rangle^{(p)}(\xi_{p-1}-\xi_p)-\tilde\gamma_j^2\EE f_1\dots f_{k+1}\langle\sigma\rangle\sum_{p\geq j+1}^{k+1}\langle\sigma\rangle^{(p)}(\xi_{p-1}-\xi_p)\nonumber\\
        &=-\tilde\gamma_j^2(\xi_j-\xi_{j-1})\EE f_1\dots f_{k+1}\langle\sigma\rangle\langle\sigma\rangle^{(j)}=
        -\tilde\gamma_j^2(\xi_j-\xi_{j-1})\EE f_1\dots f_j(\langle\sigma\rangle^{(j)})^2.
    \end{align}
    In the last step we used the fact that $\langle\sigma\rangle^{(j)}$ is independent on the r.v.'s $\eta_{j+1},\dots,\eta_{k+1}$. Putting the two contributions together the proof is complete.
\end{proof}

The lemma that follows characterizes stationary pairs.

\vspace{3pt}
\begin{lemma}\label{lem:KKT_stationary}
    Consider $(\bar x,\bar\xi)\in\mathcal{M}_k$ a $k$-stationary pair, such that all the $\bar\xi_j$'s are different without loss of generality. Then the following consistency equations hold
    \begin{align}
        \bar x_j=\EE\prod_{p=1}^jf_p(\langle\sigma\rangle^{(j)})^2
    \end{align}for all $j=1,\dots,k$.
\end{lemma}
\begin{proof}
    The boundary of the optimization set of the $x$'s is identified by the various possible matchings $x_j=x_{j-1}$. Select a $j=2,\dots,k$. If $\bar x_j$ realizes an infimum point, it must satisfy the following conditions:
    \begin{align}\label{eq:KKT_x}
    \begin{cases}
        \left.\partial_{x_j}\mathcal{P}(x,\xi)\right|_{(\bar x,\bar\xi):\bar x_{j}=\bar x_{j-1}}\geq 0\\
        (\bar x_{j}-\bar x_{j-1})\partial_{x_j}\mathcal{P}(\bar x,\bar\xi)=0
    \end{cases}\,.
    \end{align}
    In the hypothesis $\bar x_{j-1}$ has collapsed on $\bar x_j$ one must also add
    \begin{align}
        \left.\partial_{x_{j-1}}\mathcal{P}(x,\xi)\right|_{(\bar x,\bar\xi):\bar x_{j-1}=\bar x_{j}}\leq 0\,.
    \end{align}
    If $\bar x_{j}>\bar x_{j-1}$ then the statement is automatically proved. 
    
    Say instead $0=\bar x_{j}-\bar x_{j-1}$. From the third of the above equations, using \eqref{eq:Parisi-derivative}, we infer that
    \begin{align}
        \bar x_{j-1}=\bar x_{j}\leq \EE\prod_{p=1}^{j-1}f_p(\langle\sigma\rangle^{(j-1)})^2=\EE\prod_{p=1}^{j-1}f_p(\EE_{j-1} f_j\langle\sigma\rangle^{(j)})^2\leq \EE\prod_{p=1}^{j}f_p(\langle\sigma\rangle^{(j)})^2
    \end{align}where we used \eqref{eq:bracket-recursion} and Jensen inequality. The only way for the above to be compatible with the first condition in \eqref{eq:KKT_x} is to have $\left.\partial_{x_j}\mathcal{P}(x,\xi)\right|_{(\bar x,\bar\xi):\bar x_{j}=\bar x_{j-1}}= 0$. Therefore, even if $\bar x_j$ is an extremal point on the boundary, it must still satisfy the fixed point equations.
\end{proof}

\begin{lemma}\label{lem:gammader}
Let $\mu\in\mathrm{Pr}_{k'}$ and assume that its associated pair $(x^{\mu},\xi^{\mu})\in\mathcal{M}_{k}$ is $k$-stationary, then one has 
\be\label{derivativegamma1}
\dfrac{\partial}{\partial{\gamma_\ell}} \mathcal{P}(x^\mu,\xi^\mu)=\begin{cases}-\gamma_{\ell}\sum_{j\in K_{\ell}}(\xi^{\mu}_j-\xi^{\mu}_{j-1})\, (x_j^\mu)^2\,,\quad\mathrm{if}\quad \ell=1,\ldots,r-1\\
\gamma_r\Big(1-\sum_{j\in K_r}(\xi^{\mu}_{j}-\xi^\mu_{j-1})(x_j^{\mu})^2\Big)
\,,\quad\mathrm{if}\quad \ell=r
\end{cases}\,.
\ee

Equivalently denoting by $\rho_{\mu}$ the law of the the synchronized pair of random vector defined in \eqref{def:varC} one has 
\be\label{derivativegamma2}
\dfrac{\partial}{\partial{\gamma_\ell}} \bar{\mathcal{P}}(\mu)=
\begin{cases}-\gamma_{\ell}(\zeta_\ell-\zeta_{\ell-1})\int x^2 \rho_\mu(dx\, |\,\Gamma=\gamma_{\ell})\,,\quad\mathrm{if}\quad \ell=1,\ldots,r-1\\
\gamma_r\Big(1-(1-\zeta_{r-1})\int x^2 \rho_\mu(dx\, |\,\Gamma=\gamma_{r})\Big)
\,,\quad\mathrm{if}\quad \ell=r
\end{cases}\,.
\ee

\end{lemma}
\begin{proof}
The  proof  of \eqref{derivativegamma1} proceeds as the one of Lemma \ref{lem:derivative_Parisi_x}, with the difference that we exploit the fact that $(x^\mu,\xi^\mu)$ is a stationary pair. It is indeed easy to verify through Karush–Kuhn–Tucker  conditions that even if some $x$'s lie at the boundary, e.g. $x^\mu_{j+1}=x^\mu_j$ for some $j$'s, they are still stationary points for the Parisi functional, namely \eqref{eq:Parisi-derivative} always equals zero for a stationary pair.
The equivalence with \eqref{derivativegamma2} is  proved by \eqref{eq:explicit_conditional}. 
\end{proof}

\subsection{Proof of Proposition 
\ref{continuosex}} \label{sec:Lipschitz_proof}
Let us start with a simple observation.
If $(x,\xi)\in\mathcal{M}_{k}$ is such that $\xi_j=\xi_{j-1}$ for some $j\leq k$, then 
$\mathcal{P}(x,\xi)\equiv\mathcal{P}(x,\xi)^-$
where  $(x,\xi)^-=(x,\xi)\setminus(x_j,\xi_j)\in\mathcal{M}_{k-1}$. On the other hand assume that $x_j=x_{j+1}$ for some $j$, then there are two cases: if $\tilde{\gamma_j}=\tilde{\gamma}_{j+1}$ then
$\mathcal{P}(x,\xi)=\mathcal{P}(x,\xi)^-$ where $(x,\xi)^-$ is again obtained dropping $x_{j}$ and $\xi_{j}$. 

On the other hand, if $\tilde{\gamma}_{j+1}>\tilde{\gamma}_{j}$, that occurs when $\xi_{j}=\zeta_{\ell}$ for some $\ell\in\{0,\ldots,r\}$, and still $x_{j+1}=x_j$, then none of the $\xi$ can be dropped. Whenever repetitions of the $x$'s occur in correspondence of an index $j$ such that $\xi_j\in\zeta$ the pair $(x,\xi)$ cannot be simplified. We thus call a pair $(x,\xi)$ \emph{minimal} iff it has only such repetitions in the sequence $x$. As noticed in Section \ref{sec:synchro} minimal pairs are precisely the ones generated from discrete distributions $\mu\in\mathrm{Pr}_k'$.

Let us now move to the proof of the Lipschitz continuity of $\bar{\mathcal{P}}$. For $\alpha=1,2$ take $\mu_{\alpha}=(y^{\alpha},m^{\alpha})\in\mathrm{Pr}_{k_{\alpha}}$ and  denote by $(x^{\alpha},\xi^{\alpha})$ the associated sequence in $\mathcal{M}$. As mentioned earlier, such sequences $(x^{\alpha},\xi^{\alpha})$ are minimal.

Let $\xi=(\xi_j)_{j\leq k^*+1}$ be the strictly increasing sequence obtained by ordering $\xi^{1}\cup\xi^{2}$  where $k^*+1=|\xi^{1}\cup\xi^{2}|$. For $\alpha=1,2$ define also an increasing sequence $\tilde{x}^{\alpha}=(\tilde{x}^{\alpha}_j)_{j\leq k^*+1}$ where
\be
\tilde{x}^{\alpha}_j=\mu_{\alpha}^{-1}(\xi_j)\,,\quad \tilde x^\alpha_{k^*+1}=1\,.
\ee
From the above definition, it is clear that the image of $\mu_{\alpha}^{-1}$ must be contained in $x^\alpha$. Therefore the above operation introduces yet again other repetitions in the $\tilde x_j^\alpha$. In particular $\tilde{x}^{\alpha}_{j+1}=\tilde{x}^{\alpha}_{j}$ { implies} $\xi_{j}\in \xi^\alpha$. We also denote by $\bar{\gamma}$ the associated sequence \eqref{eq:tilde_gamma} through $\mu_\Gamma^{-1}(\xi)$.

For any $t\in[0,1]$ consider  $x(t)=(x_j(t))_{j\leq k^*+1}$  by
\be
x_j(t)=t \,\tilde{x}^1_j+(1-t)\,\tilde{x}^2_j\,.
\ee
Therefore $(x(t),\xi)\in\mathcal{M}_{k^*}$ and we can define $\phi(t)=\mathcal{P}(x(t),\xi)$. It is not difficult to check that 
\begin{align}
\phi(1)=\mathcal{P}(\tilde{x}^1,\xi)=\mathcal{P}(x^1,\xi^1)=\bar{\mathcal{P}}(\mu_1)\\
\phi(0)=\mathcal{P}(\tilde{x}^2,\xi)=\mathcal{P}(x^2,\xi^2)=\bar{\mathcal{P}}(\mu_2)\,.
\end{align}
Using formula \eqref{eq:Parisi-derivative} one obtains
\be
\phi'(t)=\sum_{j\leq k^*+1}\tilde{\gamma}^2_j(\xi_{j}-\xi_{j-1})(\tilde{x}_j^1-\tilde{x}_j^2)C_j(t)
\ee
where $|C_j(t)|\leq 2$. Recall that $\tilde{x}_j^{1}=\mu_1^{-1}(\xi_j)$ and  $\tilde{x}_j^{2}=\mu_2^{-1}(\xi_j)$ one obtains 
\be
|\phi'(t)|\leq 2\gamma^2_r\sum_{j\leq k^*+1}(\xi_{j}-\xi_{j-1})|\mu_1^{-1}(\xi_j)-\mu_2^{-1}(\xi_j)|=2\gamma^2_r\int_0^1|\mu_1^{-1}(p)-\mu_2^{-1}(p)|dp
\ee
and hence \eqref{Lipschtiz} follows. 

In order to prove \eqref{infequality} it is  enough to show that 
\be\label{infe}
\Big\{\bar{\mathcal{P}}(\mu)\,\mid\,  \mu\in\bigcup_{k\geq 1}\mathrm{Pr}_k\Big\}=\Big\{\mathcal{P}(x,\xi)\,\mid \,  (x,\xi)\in\mathcal{M}\Big\}.
\ee
Clearly because of the map   \eqref{map} we have the inclusion $\subseteq$ in \eqref{infe}. In order to prove the converse is enough to show that  for any $(x,\xi)\in\mathcal{M}_k$
there exists $k'\leq k$ and $\mu\in\mathrm{Pr}_{k'}$
such that $\mathcal{P}(x,\xi)=\bar{\mathcal{P}}(\mu)$. Clearly one can assume without loss that 
$(x,\xi)$ is minimal and then it's easy to check that the desired $\mu$ is simply the law of the random variable $Y$ identified by $\mathbb{P}(Y=x_j)=\xi_j-\xi_{j-1}$.

\subsection{Proof of Theorem \ref{thm:overlap_moments}}
In this proof we will use the same idea of \cite{diff_parisi}. We want to exploit convexity properties of the pressure $p_N$ in \eqref{eq:pressure_per_particle} with respect to each  component of  $\beta=(\beta_{\ell})_{\ell\leq r}$  where $\beta_{\ell}=\sqrt{\gamma^2_{\ell}-\gamma^2_{\ell-1}}$. Indeed $\beta_{\ell}$ plays the role of an inverse temperature in the definition \eqref{def:hamiell} of $H_N^{\ell}$. Convexity is inherited by the function
\be
\mathcal{P}(\gamma(\beta)):=\inf_{\mu\in\mathrm{Pr}}\bar{\mathcal{P}}(\mu,\gamma(\beta))
\ee
since it is the limit of the sequence of convex functions $p_N$.
We are going to show that $\mathcal{P}(\gamma(\beta))$  is differentiable in each of the $\beta_{\ell}$.
Fix some $\ell\leq r$ and denote by $\partial_{\ell}\mathcal{P}$ the subdifferential
of $\mathcal{P}$ w.r.t. $\beta_\ell$. Thanks to  convexity  it is enough to show that $\partial_{\ell}\mathcal{P}$ is a singleton. Let  $(x^{(k)},\xi^{(k)})_{k\in\mathbb{N}}$  be a sequence of $k$-stationary pairs and denote by $\mu_k\in\mathrm{Pr}_k$ the associated probability measure, namely $\mu_k$ is the law of a random variable $Y_k$ with $\mathbb{P}(Y_k=x^{(k)}_j)=\xi^{(k)}_j-\xi^{(k)}_{j-1}$. 
By Lipschitz continuity in Proposition~\ref{continuosex} we have that $\lim_{k\to\infty}\bar{\mathcal{P}}(\mu_k)=\mathcal{P}(\gamma(\beta))$ and the limit is approached monotonically. Then there exists a sequence  $\epsilon_k \to 0$ such that 
\be
0\leq  \bar{\mathcal{P}}(\mu_k,\gamma(\beta))-\mathcal{P}(\gamma(\beta))\leq \epsilon_k.
\ee
Now if $a\in\partial_{\ell}\mathcal{P}$  repeating the proof of \cite[Theorem 1]{diff_parisi}, one obtains
\be
a=\frac{\partial \bar{\mathcal{P}}}{\partial \beta_\ell}\left(\mu_k, \gamma(\beta)\right)+\mathcal{O}\left(\sqrt{\varepsilon_{k}}\right)\,.
\ee
{The above also requires a uniform bound on the second derivative of the Parisi functional w.r.t. $\beta_\ell$, that can be dealt with as in \cite{diff_parisi}.}

Recall that if $\mu_k$ is stationary then $\frac{\partial}{\partial_{\gamma_p}}\bar{\mathcal{P}}(\mu_k,\gamma(\beta))$
is given in Lemma  \ref{lem:gammader}. Then using the relation 
\be\label{composition}
\frac{\partial}{\partial \beta_\ell}=\beta_{\ell}\sum_{p\geq \ell}^r\,\frac{1}{\gamma_p}\,\frac{\partial}{\partial \gamma_{p}}
\ee
we get 
\be
a=\beta_\ell\Big\{1- \sum_{p\geq \ell}^{r} (\zeta_p-\zeta_{p-1})\int x^2 \rho_{\mu_k}(dx\, |\,\Gamma=\gamma_{p})\Big\}+\mathcal{O}\left(\sqrt{\varepsilon_{k}}\right)\,.
\ee
Now since $\bar{\mathcal{P}}$ is continuous there exists a subsequence $(\mu_{k_n})$ of $(\mu_{k})$ such that $\mu_{k_{n}}\xrightarrow{W_1} \mu^*$ and $\mathcal{P}(\gamma(\beta))=\bar{\mathcal{P}}(\mu^*,\gamma(\beta))$, namely $\mu^*$ is some Parisi measure. By definition of $\rho_\mu$
\begin{align}
    &\Big|\int x^2 \big(\rho_{\mu_k}(dx\, ,\,\Gamma=\gamma_{p})-\rho_{\mu^*}(dx\, ,\,\Gamma=\gamma_{p})\big)\Big|\leq \EE_U\big|(\mu_{k_n}^{-1}(U))^2-(\mu^{*-1}(U))^2\big|\mathbbm{1}(\mu_{\Gamma}^{-1}(U)=\gamma_p)\nonumber\\
    &\qquad\leq2 W_1(\mu_{k_n},\mu^*)\xrightarrow[]{n\to\infty}0\,,
\end{align}where we used that $\mathbbm{1}(\dots)\leq 1$ and that $\mu_{k_n}^{-1}(U),\mu^{*-1}(U) \in[0,1]$.
The limit along such subsequence then uniquely determines $a$, proving $\bar{\mathcal{P}}$ is differentiable, and
\begin{align}
    \frac{\partial}{\partial\beta_\ell} {\mathcal{P}}(\gamma(\beta))=\beta_\ell
    \Big\{1- \sum_{p\geq \ell}^{r} (\zeta_p-\zeta_{p-1})\int x^2 \rho_{\mu^*}(dx\, |\,\Gamma=\gamma_{p})\Big\}\,.
\end{align}
Using the chain rule we can finally write
\begin{align}
     \frac{\partial}{\partial\gamma_\ell} {\mathcal{P}}(\gamma(\beta))=\begin{cases}
         -\gamma_\ell(\zeta_\ell-\zeta_{\ell-1})\int x^2\rho_{\mu^*}(dx\mid\Gamma=\gamma_\ell)\,,\quad \text{if }\ell=1,\dots,r-1\\
         \gamma_r\Big(1-(1-\zeta_{r-1})\int x^2\rho_{\mu^*}(dx\mid\Gamma=\gamma_r)\Big)\,,\quad \text{if }\ell=r
     \end{cases}
\end{align}
which proves the statement when matched with the limit of \eqref{eq:gamma_der_pN}.

\subsection{Proof of Theorem \ref{thm:annealed}}\label{sec:proofthm_annealed}

Recall that in this proof we assume  $h=0$ almost surely. Let us start noticing that

\be\label{Jensenbound}
p_N\leq \log 2 +\dfrac{\gamma^2_r}{2}
\ee
uniformly in $N$. Indeed one has
\be
\log 2 +\dfrac{\gamma_r^2}{2}\,= \,\left.p_N\right|_{\zeta=\boldsymbol{1}}
\ee
for any integer $N$. Hence the inequality \eqref{Jensenbound} easily follows from a repeated applications of Jensen inequality and  the fact that $x\mapsto x^{t}$ is concave for any $t\in[0,1]$. Moreover again by Jensen inequality one has that 
\be
p_N\geq \frac{1}{N}\E\log Z_{r,N}
\ee
uniformly in $N$, where $\EE$ denotes expectation w.r.t.\ all the disorder. Now observe that the quantity  
$\frac{1}{N}\E\log Z_{r,N}$ coincides with quenched pressure on a Sherringhton-Kirkpatrick model at inverse temperature $\beta= \sqrt{2}\gamma_r$ and then \cite{Tala_vol1} one has
\be
\lim_{N\to\infty}\frac{1}{N}\E\log Z_{r,N}= \log 2 +\frac{\beta^2}{4}\,\,\iff\,\, \beta\leq 1 \,.
\ee
This implies that  if $\gamma_r^2\leq \frac{1}{2}$ then 
\be\label{quenchedannealed}
\lim_{N\to\infty} p_N= \log 2 +\dfrac{\gamma_r^2}{2}\,.
\ee
In order to prove that $\gamma_r^2\leq \frac{1}{2}$ is also a necessary condition for the equality \eqref{quenchedannealed} we start by noticing that $\log 2 +\dfrac{\gamma^2_r}{2}$ is obtained as a limiting value of the Parisi functional $\mathcal{P}(x,\xi)$ . More precisely let us define
\be
f(x_1,\ldots,x_r)= \mathcal{P}(x,\xi^*) 
\ee
where $\bxi^*=(\xi^*)_{j\leq r+1}$ is fixed trough the choice
\be\label{xi-choice}
\xi^*_0=\zeta_0<\ldots<\xi^*_{r-1}=\zeta_{r-1}<\xi^*_r=\xi^*_{r+1}=1\,.
\ee
Then it's easy to check that
\be
\lim_{\substack{x_j\,\to\, 0^+\\j\leq r}}\, f(x_1,\ldots,x_r)\, = \,\log 2 +\dfrac{\gamma^2_r}{2}\,.
\ee
Let us consider the quantity
\begin{align}
&f(x_{r})=\lim_{\substack{x_j\,\to\, 0^+\\j\leq r-1}}
f(x_1,\ldots,x_r)\\
&=\log 2 +\frac{1}{\zeta_{r-1}}\log\mathbb{E}_{\eta}\cosh^{\zeta_{r-1}}\big(\eta\gamma_r\sqrt{2x_r}\big)+\frac{\gamma_r^{2}}{2}\Big(1-2x_{r}+(1-\zeta_{r-1})x_r^2\Big)\,.
\end{align}
Recall that the RSB bound (Proposition 3.1 in \cite{contucci-mingione}) implies that
\be\label{Guerrbou}
\lim_{N\to\infty} p_N\leq \inf_{x_r\in[0,1]} f(x_r)\,.
\ee
We want to prove that if $\gamma^2_r>\frac{1}{2}$ then  $\inf_{x_r} f(x_r)<\lim_{x_r\to 0^+} f(x_r)=\log 2+\frac{\gamma^2_r}{2}$. Equation \eqref{eq:Parisi-derivative} yields
\be
\frac{d}{dx_r}f(x_r)=(1-\zeta_{r-1})\frac{\gamma^2_r}{2}\Big(x_r-\dfrac{\mathbb{E}_{\eta} \cosh^{\zeta_{r-1}}\big(\eta\gamma_r\sqrt{2x_r}\big)\tanh^2\big(\eta\gamma_r\sqrt{2x_r}\big)}{\mathbb{E}_{\eta}\cosh^{\zeta_{r-1}}\big(\eta\gamma_r\sqrt{2x_r}\big)}\Big)\,.
\ee
This implies that $\lim_{x_r\to 0^+}\frac{d}{dx_r}f(x_r)=0$ and also that 
\be
\lim_{x_r\to 0^+}\frac{d^2}{dx_r^2}f(x_r)=(1-\zeta_{r-1})\frac{\gamma^2_r}{2}(1-2\gamma_r^2)\,.
\ee
Hence, if $\gamma_r^2>\frac{1}{2}$ one has that $\lim_{x_r\to 0^+}\frac{d^2}{dx_r^2}f(x_r)<0$.
From \eqref{Guerrbou} we can thus conclude
\be
\lim_{N\to\infty} p_N\leq \inf_{x_r\in[0,1]}f(x_r)<\lim_{x_r\to 0^+ }f(x_r)= \log 2+\frac{\gamma_r^2}{2}\,.
\ee

\subsection{Proof of Theorem \ref{thm:minRSB}}\label{sec:proofthm_minRSB}
The proof leverages on the fact that there are at least $r$ different values of $\gamma_\ell$'s, and that the stationary sequences $(x^*_j)_{j\leq k}$ of the Parisi functional are naturally ordered.

Suppose, by contradiction, that the infimum of the variational formula \eqref{eq:var_principle} is attained on a $k$-stationary pair $(x^*,\xi^*)$ such that 
\begin{align}
    0=x_0^*<x_1^*<x_2^*<\dots<x_{j_\ell}^*=x_{j_\ell+1}^*<x_{j_\ell+2}^*<\dots<x_{k}^*<x_{k+1}^*=1
\end{align}where we stress that $x_{j_\ell}^*=x_{j_\ell+1}^*$ collapsed. The index $j_\ell$ is chosen in such a way that $\xi^*_{j_\ell}=\zeta_\ell$. Hence the CDF associated to the above choice of $\mu^*$,  does not include $\zeta_\ell$ in its image, as suggested by the limiting procedure outlined in \figurename\,\ref{fig:jump}. Our assumption then entails
\begin{align}
    \inf_{(x,\xi)\in\mathcal{M}}\mathcal{P}(x,\xi)=\mathcal{P}(x^*,\xi^*)\,.
\end{align}
Thanks to Lemma~\ref{lem:KKT_stationary} we have that
\begin{align}\label{eq:gradient_x_jl}
    \left.\partial_{x_{j_\ell}}\mathcal{P}(x,\xi)\right|_{(x^*,\xi^*)}=\tilde\gamma^2_{j_\ell}(\xi^*_{j_\ell}-\xi^*_{j_\ell-1})\big[
    x_{j_\ell}^*-\EE\prod_{p=1}^{j_\ell} f_p\big(\langle\sigma\rangle^{(j_\ell)}\big)^2
    \big]=0\,.
\end{align}
\begin{figure}[h!!!]
    \centering
    \begin{tikzpicture}
\begin{axis}[
    xlabel={x},
    ylabel={$F_\mu$},
    xmin=0, xmax=1,
    ymin=0, ymax=1,
    scale=1
]
\addplot[color=blue,very thick,domain=0:0.3, 
    samples=100] {.1} node[above,pos=.5] {$\xi_0$};
    \addplot[style=dashed,domain=0:1, 
    samples=100] {.1} node[above,pos=.05] {$\zeta_0$};
\addplot[color=blue,very thick,domain=0.3:0.5, 
    samples=100] {.2} node[above,pos=.5] {$\xi_1$};
    \addplot[style=dashed,domain=0:1, 
    samples=100] {.2} node[above,pos=.05] {$\zeta_1$};
\addplot[color=red,very thick,domain=0.5:0.55, 
    samples=100] {.4} node[above,pos=.2] {$\xi_2$};
    \addplot[style=dashed,domain=0:1, 
    samples=100] {.4} node[above,pos=.05] {$\zeta_2$};
    \addplot[color=blue,very thick,domain=0.55:0.7, 
    samples=100] {.6} node[above,pos=.2] {$\xi_3$};
\addplot[color=blue,very thick,domain=0.7:0.8, 
    samples=100] {.7} node[above,pos=.2] {$\xi_4$};
    \addplot[domain=0:1,style=dashed, 
    samples=100] {.7} node[above,pos=.05] {$\zeta_3$};
\addplot[color=blue,very thick,domain=0.8:0.9, 
    samples=100] {.85} node[above,pos=.2] {$\xi_5$};
\addplot[color=blue,ultra thick,domain=0.9:1, 
    samples=100] {1} node[below,pos=.2] {$\xi_6$};

\addplot[color=red,style=dashed,domain=0:0.4, 
    samples=50] (0.5,{x}) node[pos=.1,left] {$x_2$};
\addplot[color=red,style=dashed,domain=0:0.4, 
    samples=50] (0.55,{x}) node[pos=.1,right] {$x_3$};
\draw[thick, color=red,->] (0.55,0.05)--(0.5,0.05);
\end{axis}
\end{tikzpicture}
\caption{Typical limiting situation in which one of the $\zeta_\ell$'s (in this case $\zeta_2$) disappears from the cumulative distribution function. In general, $\zeta_\ell$ is not in the final limiting distribution when $x_{j_{\ell}+1}\to x_{j_\ell}$ for $j_\ell$ s.t. $\xi_{j_\ell}=\zeta_\ell$. In this plot, $r=4$, $k=6$.}\label{fig:jump}
\end{figure}

Recall that $\xi^*_{j_\ell}=\xi^*_{j_\ell-1}$ can be discarded without loss of generality. In fact, in that case the $\inf$ would be attained on a distribution with less than $k$ distinct relevant values in the sequence $x$. It is indeed intuitive from \figurename\,\ref{fig:jump}, that in order to skip $\zeta_2$, $\xi_1\,,\xi_3$ and $\xi_2$ (which equals $\zeta_2$) must be all different. 

In order to prove Theorem \ref{thm:minRSB}, we show that, under the assumption that $(x^*,\xi^*)$ is the infimum point of $\mathcal{P}$ on $\mathcal{M}_k$, we are actually able to construct a new pair $(\bar{x},\bar{\xi})$ with finite number of atoms, whose cumulative distribution contains $\zeta_{\ell}$ in the image and such that
\be\label{eq:condition}
\mathcal{P}(x^*,\xi^*)>\mathcal{P}(\bar{x},\bar{\xi})\,.
\ee

Let us design a pair $(\bar{x},\bar{\xi})$ with $k+1$ distinct atoms as follows. We take $\bar\xi=\xi^*$ and
\begin{align}
    0=x_0^*=\bar x_0<x_1^*=\bar x_1<\dots<x_{j_\ell}^*=\bar x_{j_\ell}=x_{j_\ell+1}^*<\bar x_{j_\ell+1}<\dots<x_{k+1}^*=\bar x_{k+1}=1\,.
\end{align}In broad terms the sequence $\bar x$ is the sequence $x^*$ where $x_{j_\ell}^*$ remained detatched from $x_{j_\ell+1}^*$. With these notations, following \eqref{eq:tilde_gamma}, one has
\be
\tilde\gamma_{j_\ell}=\gamma_\ell\,<\, \tilde\gamma_{j_{\ell}+1}=\gamma_{\ell+1}\,.
\ee
It remains to show that there exists a choice of $\bar x_{j_\ell+1}$ such that \eqref{eq:condition} holds true. This is verified if we prove that
\begin{align}
    \left.\partial_{x_{j_\ell+1}}\mathcal{P}(x,\xi)\right|_{(x^*,\xi^*)}=\tilde\gamma_{j_\ell+1}^2(\xi^*_{j_\ell+1}-\xi^*_{j_\ell})\big[x_{j_\ell+1}^*-\EE\prod_{p=1}^{j_\ell+1} f_p\big(\langle\sigma\rangle^{(j_\ell+1)}\big)^2
    \big]
\end{align}
is strictly negative. Using the fact that $x_{j_\ell}^*=x_{j_\ell+1}^*$, together with \eqref{eq:gradient_x_jl} and \eqref{eq:bracket-recursion} one readily gets
\begin{align}\label{eq:variance_bracket}
    \left.\partial_{x_{j_\ell+1}}\mathcal{P}(x,\xi)\right|_{(x^*,\xi^*)}=-\tilde\gamma_{j_\ell+1}^2(\xi^*_{j_\ell+1}-\xi^*_{j_\ell})\EE \prod_{p=1}^{j_\ell} f_p \big[\EE_{j_\ell}f_{j_\ell+1}\Big( \langle\sigma\rangle^{(j_\ell+1)}-\EE_{j_\ell}f_{j_\ell+1}\langle\sigma\rangle^{(j_\ell+1)}\Big)^2\big]
\end{align}
which is non positive. Keeping in mind the example of \figurename\,\ref{fig:jump}, it is not difficult to see that, in the case we are interested in, we can assume $\xi^*_{j_\ell+1}>\xi^*_{j_\ell}$ without any loss of generality. Hence, the sign of the previous derivative is uniquely determined by the variance in \eqref{eq:variance_bracket}.

Recall that
\begin{align}\label{eq:bracket_proof}
        \langle\sigma\rangle^{(j_\ell+1)}&= \EE_{j_\ell+1}\dots\EE_{k}f_{j_\ell+2}\dots f_{k+1}\tanh \Big(\sum_{p=1}^{j_\ell}\eta_p\sqrt{2\big( \tilde \gamma_{p}^2 x^*_{p}- \tilde \gamma_{p-1}^2 x^*_{p-1}\big)}\nonumber\\
        &\qquad+\eta_{j_\ell+1}\sqrt{2\big(\gamma_{\ell+1}^2- \gamma_{\ell}^2 \big)x^*_{j_\ell}}+\sum_{p=j_\ell+2}^{k+1}\eta_p\sqrt{2\big( \tilde \gamma_{p}^2 x^*_{p}- \tilde \gamma_{p-1}^2 x^*_{p-1}\big)}+h\Big)\,,
\end{align}where we have imposed $x^*_{j_\ell+1}=x_{j_\ell}^*$ and $\tilde \gamma_{j_\ell+1}=\gamma_{\ell+1}>\tilde\gamma_{j_\ell}=\gamma_\ell$. The summation of the first $j_\ell+1$ contributions in the above, and possibly the magnetic field, is what makes $\langle\sigma\rangle^{(j_\ell+1)}$ random. From the previous formula we can see that, as long as $\eta_{j_\ell+1}$ appears, $\langle\sigma\rangle^{(j_\ell+1)}$ cannot have vanishing variance in the measure $\EE_{j_\ell}f_{j_\ell+1} (\cdot)$. This always occurs precisely because $\gamma_{\ell+1}>\gamma_\ell$, unless all the $x_j^*, \forall\, j\leq j_\ell+1$ collapse to zero. 
{ This is ruled out by the two following Lemmas.
\vspace{3pt}

\begin{lemma}
    If $\EE h^2>0$ then $x_{j}^*>0$ for all $j=1,\dots,k$.
\end{lemma}
\begin{proof}
    To prove the result, we show that at $x_{j}^*=0$ for some $j$, $\partial_{x_{j}}\mathcal{P}$ is strictly negative. Recall that
    \begin{align}
        \langle\sigma\rangle^{(j)}&= \EE_{j}\dots\EE_{k}f_{j+1}\dots f_{k+1}\tanh \Big(\sum_{p=1}^{k+1}\eta_p\sqrt{2\big( \tilde \gamma_{p}^2 x^*_{p}- \tilde \gamma_{p-1}^2 x^*_{p-1}\big)}+h\Big)\,.
    \end{align}Using a reasoning similar to that used to prove \eqref{eq:f-derivative1}-\eqref{eq:f-derivative2}, one can show
\begin{align}\label{eq:D_h_f}
    \partial_h f_j=\xi_{j-1}(\langle\sigma\rangle^{(j)}-\langle\sigma\rangle^{(j-1)}) f_j\,.
\end{align}
Then we can finally compute the derivative
\begin{align}\label{eq:Dh_sigma_j}
    \partial_h \langle\sigma&\rangle^{(j)}=\EE_{j}\dots\EE_{k}\sum_{s=j+1}^{k+1}f_{j+1}\dots\partial_hf_s\dots f_{k+1}\langle\sigma\rangle+1-\EE_{j}\dots\EE_{k}f_{j+1}\dots f_{k+1}\langle\sigma\rangle^2=\nonumber\\
    &=1-\EE_{j}\dots\EE_{k}f_{j+1}\dots f_{k+1}\langle\sigma\rangle^2+
    \sum_{s=j+1}^{k+1} \EE_{j}\dots\EE_{k}f_{j+1}\dots f_{k+1}
    \xi_{s-1}(\langle\sigma\rangle^{(s)}-\langle\sigma\rangle^{(s-1)}) \langle\sigma\rangle\,.
\end{align}The last set of terms in the above line can be rewritten using the recursive relation \eqref{eq:bracket-recursion}, and the first one is strictly positive. Hence:
\begin{align}
    \partial_h \langle\sigma&\rangle^{(j)}>\sum_{s=j+1}^{k+1}\xi_{s-1}\EE_{j}\dots\EE_{s-1}f_{j+1}\dots f_{s}
    [(\langle\sigma\rangle^{(s)})^2-(\EE_{s-1} f_s\langle\sigma\rangle^{(s)})^2]\geq 0\,.
\end{align}The above is actually true for any value of the variational parameters $\xi,x$. Suppose now, by contradiction, that $x_j^*=0$, which also entails $x_s^*=0$ for all $s\leq j$ by inequality constraints automatically satisfied by stationary points, see Lemma \ref{lem:KKT_stationary}. In that case one would also have $\langle\sigma\rangle^{(j)}|_{h=0}=0$ by symmetry of the Gaussian law of the $\eta$'s, and of the weights $f$ (recall $x^*_{s\leq j}=0$). This then means that as soon as the magnetic field is non zero $\langle\sigma\rangle^{(j)}|_{h\neq0}\neq0$. This would then make it impossible for the derivative \eqref{eq:Parisi-derivative} to vanish when evaluated at $x_j=0$. This contradiction is born from the assumption $x_j^*=0$ which is thus falsified concluding the proof.
\end{proof}

\vspace{3pt}

\begin{lemma}\label{lem:partial_annealing--}
    Consider the case $\EE h^2=0$. If $\gamma_\ell^2>\frac{1}{2\zeta_\ell^2}$ then $x^*_{j_\ell}>0$.
\end{lemma}
\begin{proof}
    The strategy is to show that $x_{j\leq j_{\ell}}^*=0$ is an unstable stationary point under the hypothesis. In order to do it, we need an expansion of $\langle\sigma\rangle^{(j_\ell)}$ around $x_{j_\ell}\to 0$, with $x_{j_\ell-1}=\dots=x_1= 0$. Assuming the first $j_\ell-1$ coordinates are $0$ is licit as for stationary points $x^*$ they are naturally ordered, as per Lemma \ref{lem:KKT_stationary}.
    
    In this selected limiting direction (which is compatible with the inequality constraints) we can write:
    \begin{align}
    \langle\sigma\rangle^{(j_\ell)}= \EE_{j_\ell}\dots\EE_{k}f_{j_\ell+1}\dots f_{k+1}\tanh \Big(\sqrt{2x_{j_\ell}}\eta_{j_\ell}\gamma_\ell +\sum_{p=j_\ell+1}^{k+1}\eta_p\sqrt{2(\tilde\gamma_p^2 x_p-
    \tilde\gamma_{p-1}^2 x_{p-1})}\Big)\,,
    \end{align}where we have kept the remaining $x$'s generic, and $\tilde\gamma_{j_\ell}=\gamma_\ell$. Even though it was not specified, a similar separation can be carried out inside the $f$'s. The term $\sqrt{2x_{j_\ell}}\eta_{j_\ell}\gamma_\ell$ behaves as a random magnetic field. Hence we can expand around it using the formula for the derivative \eqref{eq:Dh_sigma_j}:
    \begin{align}
    \langle\sigma&\rangle^{(j_\ell)}= \sqrt{2x_{j_\ell}}\eta_{j_\ell} \gamma_\ell \EE_{j_\ell}\dots\EE_{k}f^0_{j_\ell+1}\dots f^0_{k+1}\big(1-\langle\sigma\rangle^2\big) \nonumber\\
    &+\sqrt{2x_{j_\ell}} \eta_{j_\ell}\gamma_\ell\sum_{p=j_\ell+1}^{k+1}\xi_{p-1}\EE_{j_\ell}\dots\EE_{k}f^0_{j_\ell+1}\dots f^0_{k+1}\big(\langle\sigma\rangle^{(p)}-\langle\sigma\rangle^{(p-1)}\big)\langle\sigma\rangle+\eta_{j_\ell}^2O(x_{j_\ell})\,,
    \end{align}where the superscript $^0$ signals that $x_{j_\ell}=0$ in those $f$'s, and for the sake of presentation we have omitted that, inside all brackets $\langle\sigma\rangle^{(\dots)}$, $x_{j\leq j_\ell}$ have been set to $0$ too. The $O(x_{j_\ell})$ comes from the fact that the second derivative w.r.t.\ an external field is uniformly bounded, thanks to \eqref{eq:D_h_f}, and the fact that derivatives of $\tanh$ are uniformly bounded. The $\eta^2_{j_\ell}$ (and its powers) will be later averaged under a Gaussian standard measure yielding a controllable constant. Notice also that $f_j^0=1$ for all $j\leq j_\ell$. Therefore
    \begin{align}
        \EE\prod_{s=1}^{j_\ell} f_s(\langle\sigma \rangle^{(j_\ell)})^2&=2x_{j_\ell} \gamma_\ell^2\big[\EE_{j_\ell}\dots\EE_{k}f^0_{j_\ell+1}\dots f^0_{k+1}\big(1-\langle\sigma\rangle^2\big)\nonumber\\
        &+\sum_{p=j_\ell+1}^{k+1}\xi_{p-1}\EE_{j_\ell}\dots\EE_{k} f^0_{j_\ell+1}\dots f^0_{k+1}\big(\langle\sigma\rangle^{(p)}-\langle\sigma\rangle^{(p-1)}\big)\langle\sigma\rangle\big]^2+O(x_{j_\ell}^{3/2})\,.
    \end{align}
    For all the $p$'s in the summation one has $\xi_{p-1}\geq\xi_{j_\ell}=\zeta_\ell$. Furthermore, as we have already observed for \eqref{eq:Dh_sigma_j}, the last summation $\sum_{p=j_\ell+1}^{k+1}\dots$ contains only non-negative terms. Following these considerations the previous equation turns into
    \begin{align}
        \dots &=2x_{j_\ell}\gamma_\ell^2\big[\EE_{j_\ell}\dots\EE_{k}f^0_{j_\ell+1}\dots f^0_{k+1}\big(1-(1-\zeta_\ell)\langle\sigma\rangle^2+\zeta_\ell\langle\sigma\rangle^{(j_\ell)}\langle\sigma\rangle\big)\big]^2+O(x_{j_\ell}^{3/2})\,,
    \end{align}where we have solved the last telescopic sum. We stress again that $x_{j\leq j_\ell}=0$ in all the above brackets on the r.h.s. This implies that $\langle\sigma\rangle^{(j_\ell)}=0$ again by symmetry of the laws of the $\eta$'s and their tilts $f$. Hence, finally
    \begin{align}
        \EE\prod_{s=1}^{j_\ell} f_s(\langle\sigma \rangle^{(j_\ell)})^2&\geq2x_{j_\ell}\gamma_\ell^2\zeta_\ell^2(1+o(1))\,,
    \end{align}where we just used $\langle\sigma\rangle^2\leq 1$ and $o(1)$ vanishes as $x_{j_\ell}\to0$. This can be used to bound the derivative of the Parisi functional when $x_{j_\ell}\to 0$:
    \begin{align}
        \partial_{x_{j_\ell}}\mathcal{P} \leq\gamma_\ell^2(\zeta_\ell-\xi_{j_\ell-1})x_{j_\ell}\big[
        1-2\gamma_\ell^2 \zeta_\ell^2+o(1)
        \big]\,.
    \end{align}
    The above also yields a bound on the second derivative at $x_{j
    \leq j_\ell}=0$ in direction $x_{j_\ell}$:
    \begin{align}
    \partial^2_{x_{j_\ell}} \mathcal{P}|_{x_{j\leq j_\ell}=0} \leq\gamma_\ell^2(\zeta_\ell-\xi_{j_\ell-1})\big[
    1-2\gamma_\ell^2 \zeta_\ell^2
    \big]\,,
    \end{align}which is negative whenever $\gamma_\ell^2>\frac{1}{2\zeta_\ell^2}$. Therefore the stationary point $x_{j\leq j_\ell}=0$ is unstable precisely in direction $x_{j_\ell}$ proving the claim.

\end{proof}

}

To summarize, we have shown that, under the hypothesis \eqref{eq:low-T-condition} or $h\neq 0$ with positive probability
\begin{align}
    \mathcal{P}(x^*,\xi^*)>\mathcal{P}(\bar x,\bar\xi)\,.
\end{align}This concludes the proof by absurd.

\subsection{Proof of Theorem \ref{thm:plateau}}

In this section we denote by $\mu$   a solution of the variational problem \eqref{infequality} which is  is the $W_1$-limit of a sequence $(\mu_n)$ associated to some $k_n$-stationary pair. We start proving the following

\vskip 0.3cm
\begin{proposition}\label{prop:plateau}
If $\zeta_{\ell}\,\gamma^2_{\ell+1}\,< 1/2$ then 
 \be\label{iceplat}
\Delta_{\ell}(\mu)\geq 2\,\left(\gamma^2_{\ell+1}-\gamma^2_{\ell}\right)\,\limsup_{n\to\infty}\mu_n^{-1}(\zeta^+_{\ell})\, \Big(\int^1_ {\mu^{-1}(\zeta^+_{\ell})}\mu([0,x])\,dx \Big)^2 \,.
\ee
\end{proposition}
The above inequality still contains a limit. The latter can be removed at the expense of a slightly weaker inequality, that however offers the advantage of involving the limiting distribution $\mu$ only. Indeed one can define  the following smoothing of the quantile function:
\begin{align}
    \nu_{\delta,\ell}^{-1}:=\frac{1}{\delta}\int_{\zeta_\ell}^{\zeta_\ell+\delta}\nu^{-1}(p)dp\,,
\end{align}for any $\delta>0$. As an immediate consequence of its definition we have that
\be\label{bounds} 
\begin{aligned}
&\nu^{-1}(\zeta_{\ell}^+)\leq\nu_{\delta,\ell}^{-1}\leq \nu^{-1}(\zeta_{\ell}+\delta)\\
&\nu^{-1}(\zeta_{\ell}-\delta)\leq\nu_{-\delta,\ell}^{-1}\leq \nu^{-1}(\zeta_{\ell})\,.
\end{aligned}
\ee
Thanks to the above equalities we will be able to exploit the $W_1$ convergence of $\mu_n$ to $\mu$. For instance, we can prove the following

\vspace{3pt}
\begin{lemma}
\be\label{cristoff}
\mu_{n}\xrightarrow{W_1} \mu\,\Rightarrow\, \Delta_{\ell}(\mu)\geq \limsup_{n\to\infty}\Delta_{\ell}(\mu_n).
 \ee
\end{lemma}

\begin{proof}
For any non empty interval $A\subseteq[0,1]$ one has 
\be\label{W1integral}
\mu_{n}\xrightarrow{W_1}\mu\,\Rightarrow\,\lim_{n\to\infty}\Big|\int_A \left(\mu^{-1}(p)-\mu_n^{-1}(p)\right)\, dp\Big|=0.
\ee
The above implies
\be\label{integralaverages}
\mu_{\delta,\ell}^{-1}=(\mu_n^{-1})_{\delta,\ell}+o_n(1)\,,\quad\mu_{-\delta,\ell}^{-1}=(\mu_n^{-1})_{-\delta,\ell}+o_n(1)
\ee for any positive $\delta$.
Therefore, using \eqref{bounds}, we get
\be
\mu^{-1}(\zeta_{\ell}+\delta)-\mu^{-1}(\zeta_{\ell}-\delta)\geq 
 \mu_{\delta,\ell}^{-1}-
\mu_{-\delta,\ell}^{-1}=
(\mu_n^{-1})_{\delta,\ell}-
(\mu_n^{-1})_{-\delta,\ell}+o_n(1)\geq\mu^{-1}_n(\zeta^+_{\ell})-\mu^{-1}_n(\zeta_{\ell})+o_n(1).
\ee
Take first $n\to\infty$ and then $\delta\to0^+$ one get \eqref{cristoff}.
\end{proof}

    \label{rem:limiting_inequality}
    From \eqref{bounds} we observe that
    \begin{align}
        \mu^{-1}(\zeta_\ell+\delta)\geq \mu^{-1}_{\delta,\ell}=(\mu^{-1}_n)_{\delta,\ell}+o_n(1)\geq \mu^{-1}_n(\zeta_\ell^+)+o_n(1)\,.
    \end{align}Hence, by sending $n\to\infty$ and then $\delta\to 0$ we obtain
    \begin{align}\label{eq:ineq_limsup_quantile}
        \mu^{-1}(\zeta^+_\ell)\geq \limsup_{n\to\infty}\mu^{-1}_n(\zeta_\ell^+)\,.
    \end{align}    
    On the other hand, if one is willing to sacrifice the right limit,
    \begin{align}
        \limsup_{n\to\infty}\mu^{-1}_n (\zeta_\ell^+)\geq \limsup_{n\to\infty}(\mu^{-1}_n)_{-\delta,\ell}=\mu^{-1}_{-\delta,\ell}
    \end{align}for any $\delta>0$. By sending $\delta\to0$, and left continuity of the quantile, we get
    \begin{align}\limsup_{n\to\infty}\mu^{-1}_n (\zeta_\ell^+)\geq \mu^{-1}(\zeta_\ell).
    \end{align} Used in \eqref{iceplat} the above leads to \eqref{iceplatthem}. Let's go back to the proof of Proposition \eqref{prop:plateau}. We will prove at first the following finite $k$ version of it:

\vspace{3pt}
\begin{proposition}\label{superestimate}
Given $k\in \mathbb{N}$, let $\mu_k$ be a measure associated to a $k$-stationary pair. If $\,\zeta_{\ell}\,\gamma^2_{\ell+1}\,< 1/2$ then 
\be\label{finitekappa}
\mu_k^{-1}(\zeta_{\ell}^+)-\mu_k^{-1}(\zeta_{\ell}) \geq 2\,\mu_k^{-1}(\zeta^+_{\ell})\big(\gamma^2_{\ell+1}-\gamma^2_{\ell}\big)\,\left(\int^1_{\mu^{-1}_k(\zeta_{\ell}^+)} \mu_k([0,x])\,dx \right)^2 \,.
\ee
\end{proposition}
\noindent Assume for the moment that the above proposition holds, then  \eqref{iceplat} also holds. In fact, it suffices to use \eqref{finitekappa} and \eqref{cristoff} together with the continuity of the integral w.r.t. its integration extremes and \eqref{eq:ineq_limsup_quantile}.

Let us start recalling a classical result. 

\vspace{3pt}

\begin{lemma}(Cramer-Rao bound)\label{Cramerrao}
Let $\phi\in C^2(\mathbb{R},\mathbb{R})$ and consider the probability density $p(x) \propto e^{-\phi(x)}$ then for $h$ locally Lipscthiz one has

\be\label{eq:cramerrao}
\mathrm{Var}(h)\geq \left[\mathbb{E} (h')\right]^2 \left[\mathbb{E} \left(\phi''\right) \right]^{-1}.
\ee
\end{lemma}
\noindent For a proof see \cite{Cramer_Rao_Bound}.

In order to use the above, we pick a $k$-stationary  pair $(x,\xi)$ and assume without loss that it is minimal. Fix some $\ell< r$ and let $j_{\ell}\in\{1,\ldots, k\}$ be the unique index such that $\tilde\gamma_{j_\ell}=\gamma_\ell\,<\, \tilde\gamma_{j_{\ell}+1}=\gamma_{\ell+1}$ or equivalently $\xi_{j_\ell}=\zeta_{\ell}$. By Lemma \ref{lem:KKT_stationary} one has
\be\label{difference}
x_{j_{\ell+1}}-x_{j_{\ell}}= \E \prod_{p=1}^{j_{\ell}} f_p \,\left[\mathrm{Var}_{j_{\ell}+1}\left(\langle\sigma \rangle^{(j_\ell+1)} \right)\right]
\ee
where, conditionally on $(\eta_j)_{j\leq j_{\ell}}$, the quantity $\mathrm{Var}_{j_{\ell}+1}$ is the variance w.r.t. the 
average $\E_{j_\ell} f_{j_\ell+1}$, that is associated with the probability density
\be\label{prob density}
p(\eta_{j_{\ell}+1})\propto \exp\left[-\frac{1}{2}\eta^2_{j_\ell+1}+\xi_{j_\ell} \log \tilde{Z}_{j_\ell+1} \right].
\ee
The idea is now to use Lemma~\ref{Cramerrao} to bound $\mathrm{Var}_{j_{\ell}+1}\left(\langle\sigma \rangle^{(j_\ell+1)}\right)$ conditionally on $(\eta_j)_{j\leq j_{\ell}}$. In order to lighten the notation we set $j_{\ell}+1\equiv i$ and $\eta_{j_\ell+1}\equiv \eta$. Therefore by \eqref{eq:cramerrao} we have
that
\be\label{eq:137}
\mathrm{Var}_i(\langle \sigma\rangle^{(i)}) \geq \left[\mathbb{E}_{\eta} \left(\frac{\partial h}{\partial \eta}\right)\right]^2 \left[\mathbb{E}_{\eta} \left(\frac{\partial^2 \phi}{\partial \eta^2}\right) \right]^{-1}
\ee
where 
\be\label{defphiacca}
\begin{aligned}
\eta&\mapsto \phi(\eta)= \frac{1}{2}\eta^2-\xi_{i-1} \log \tilde{Z}_{i}(\eta)\\
\eta &\mapsto h(\eta)=\langle \sigma \rangle^{(i)}  \end{aligned}
\ee
and $\E_{\eta}$ is the average w.r.t. the probability density \eqref{prob density}. Recall that since $\xi_k=\xi_{k+1}=1$ we have the following simplification:
\begin{align}
    \langle\sigma\rangle^{(i)}=\EE_{\geq i}\prod_{p=i+1}^kf_p\tanh\big( \sum_{j=1}^k\eta_j\sqrt{2(\tilde\gamma_j^2x_j-\tilde\gamma_{j-1}^2x_{j-1})}
    \big)
\end{align}namely $f_{k+1}$ can be explicitly integrated away.

\vskip 0.3cm

\begin{lemma}\label{estimatesstaircase}
Let $h(\eta)$ and $\phi(\eta)$ the random functions defined in \eqref{defphiacca}. If $\zeta_{\ell}\,\gamma^2_{\ell+1}\,< 1/2$ then  
\be\label{hessianestima}
0<\mathbb{E}_{\eta} \left(\frac{\partial^2}{\partial \eta^2}\phi\right) \leq 1\,.
\ee
Moreover,
\be\label{derivativeacca}
\E_{\eta}\left(\frac{\partial h}{\partial \eta} \right)= \sqrt{2(\gamma^2_{\ell+1}x_i-\gamma^2_{\ell}x_{i-1})}\,\E_{\eta}\left(\,\sum_{p=i+1}^{k+1} \xi_{p-1} (a_p-a_{p-1})\right)
\ee
where, for any $p\leq k$
\be\label{eq:a_p_def}
a_p\,:=\,\E^{<p}_{\geq i}\prod_{s=i+1}^p f_s \left(\langle \sigma \rangle^{(p)}\right)^2\,,\text{ and }a_{k+1}=1\,.
\ee
\end{lemma}

\begin{proof}
Let us start recalling that $\tilde{Z}_{i}(\eta)$ is obtained recursively starting from $\tilde{Z}_{k+1}=2\cosh(z)$ where $z=\sum_{j\leq k+1} \eta_{j}c_j$ and  $c_j={\sqrt{2(\tilde{\gamma}^2_{j} x_{j}-\tilde{\gamma}^2_{j-1} x_{j-1})}}$.

Therefore

\be
\begin{aligned}
\frac{\partial}{\partial \eta}{\phi}(\eta)&= \eta-\xi_{i-1} \frac{\partial}{\partial \eta}\log \tilde{Z}_{i}(\eta)\\
&=\eta-\xi_{i-1} \E_{\geq i}\prod_{p=i+1}^{k+1} f_p  \frac{\partial}{\partial \eta}\log \tilde{Z}_{k+1}(\eta)\\
&=\eta-\xi_{i-1} c_i \langle \sigma \rangle^{(i)}=  \eta-\xi_{i-1} c_i h(\eta)
\end{aligned}
\ee
and 
\be
\frac{\partial^2}{\partial \eta^2}{\phi}(\eta)=1-\xi_{i-1} c_i \frac{\partial}{\partial \eta}h(\eta).
\ee
Now $\frac{\partial}{\partial \eta}h(\eta)$ can be computed as follows
\be\label{AplusB}
\begin{aligned}
\frac{\partial}{\partial \eta}h(\eta)&= \E_{\geq i} \frac{\partial}{\partial \eta}\Big[\prod_{p=i+1}^{k}f_p \tanh\big(\sum_{j=1}^k\eta_jc_j\big)\Big]= \E_{\geq i} \sum_{p=i+1}^{k} \frac{\partial}{\partial \eta} [f_p] \prod_{s=i+1,s\neq p}^{k}f_s \tanh \big(\sum_{j=1}^k\eta_jc_j\big)\\
&+ \E_{\geq i}  \prod_{p=i+1}^{k}f_p \frac{\partial}{\partial \eta}[\tanh\big(\sum_{j=1}^k\eta_jc_j\big)]=A+B
\end{aligned}
\ee 
where $A$ and $B$ are defined as the first and second term respectively on the left hand side of the above.

Recall that $\eta\equiv\eta_i$ then one can use relation \eqref{eq:f-derivative2} to obtain
\be\label{stepover}
A=c_i \,\sum_{p=i+1}^{k} \xi_{p-1} \, \E_{\geq i} \,\prod_{s=i+1}^{k} f_s \,\tanh \big(\sum_{j=1}^k\eta_jc_j\big)\,\left(\langle \sigma \rangle^{(p)}-\langle \sigma \rangle^{(p-1)}\right).
\ee

Keeping in mind that  $\langle \sigma \rangle^{(p)}=\E_{\geq p}\prod_{s=p+1}^{k+1} f_s \tanh (z)=\E_{\geq p}\prod_{s=p+1}^{k} f_s \tanh (\sum_{j=1}^k\eta_jc_j)$ depends on $(\eta_{j})_{j\leq p}$ and $f_s$ depends on $(\eta_{j})_{j\leq s}$ one can write for any $p$-term in the previous line
\be\label{stepunder}
\begin{aligned}
&\E_{\geq i} \,\prod_{s=i+1}^{k} f_s \,\tanh \big(\sum_{j=1}^k\eta_jc_j\big)\,\left(\langle \sigma \rangle^{(p)}-\langle \sigma \rangle^{(p-1)}\right)=
\E^{<p}_{\geq i} \,\prod_{s=i+1}^{p} f_s\, 
(\langle \sigma \rangle^{(p)})^2-\E^{<p-1}_{\geq i} \,\prod_{s=i+1}^{p-1} f_s 
(\langle \sigma \rangle^{(p-1)})^2.
\end{aligned}
\ee
Using \eqref{stepunder} in \eqref{stepover} one gets
\be
A=c_i \,\sum_{p=i+1}^{k} \xi_{p-1} (a_p-a_{p-1})
\ee with $a_p$ as in \eqref{eq:a_p_def}.

The second term in \eqref{AplusB} is 
\be
B=c_i(1-\E_{\geq i}\prod_{s=i+1}^{k} f_s \tanh^2\big(\sum_{j=1}^k\eta_jc_j\big)=c_i(1-a_{k}).
\ee
Therefore from \eqref{AplusB} one gets
\be
\frac{\partial}{\partial \eta}h(\eta)=A+B=c_i \,\sum_{p=i+1}^{k+1} \xi_{p-1} (a_p-a_{p-1}) 
\ee
where we used $a_{k+1}=1, \,\xi_k=1$. Taking the expectation on both sides we prove \eqref{derivativeacca}.

Therefore, since $1\geq a_p\geq a_{p-1}\geq 0$ for all $p$'s, one has
\begin{align}
    0\leq \frac{\partial}{\partial \eta}h(\eta)\leq c_i\sum_{p=i+1}^{k+1}(a_{p}-a_{p-1})=c_i(1-a_i)\leq c_i\,.
\end{align}
and hence
\be\label{bound derivatives}
1- \xi_{i-1}c^2_i\leq \frac{\partial^2}{\partial \eta^2}{\phi}(\eta)\leq 1
\ee
which is uniform in the length $k$ of the stationary pair $(x,\xi)$ . In particular \eqref{bound derivatives} implies that
\be\label{porcazozza}
1-\xi_{i-1}c^2_i\leq
\mathbb{E}_{\eta} \left(\frac{\partial^2}{\partial \eta^2}\phi\right)\leq  1.
\ee

Keeping in mind that $c_i^2=2(\gamma^2_{\ell+1}x_{j_{\ell+1}}-\gamma^2_{\ell}x_{j_{\ell}})$, we have that $c_i^2\leq 2\gamma^2_{\ell+1}$ and  $1- \xi_{i-1}c^2_i\geq 1-2\xi_{i-1}\gamma^2_{\ell+1}$. Therefore the condition $1-2\xi_{i-1}\gamma^2_{\ell+1}>0$
implies that
\be\label{porcamiseria}
0<\mathbb{E}_{\eta} \left(\frac{\partial^2}{\partial \eta^2}\phi\right) \leq 1\,.
\ee

\end{proof}

\begin{proof}[Proof of Proposition~\ref{superestimate}]
Thanks to the previous Lemma if $1-2\zeta_{\ell}\gamma^2_{\ell+1}>0$  one can use \eqref{hessianestima}, equation \eqref{derivativeacca} and Jensen inequality  in the relations \eqref{difference} {and \eqref{eq:137}} obtaining
\be
x_{i}-x_{i-1}\geq 2\left(\gamma^2_{\ell+1}x_i-\gamma^2_{\ell}x_{i-1}\right)\,\Big[\E\,\prod_{s=1}^{i-1} f_s \, \E_{\eta}\Big(\,\sum_{p=i+1}^{k+1} \xi_{p-1} (a_p-a_{p-1}) \Big ) \Big]^2.
\ee
Now by stationary conditions in Lemma \ref{lem:KKT_stationary}  one has for all $p$'s

\be\
\E\,\prod_{s=1}^{i-1} f_s \, \E_{\eta} \,a_p=\E \,\prod_{s=1}^p f_s(\langle\sigma\rangle^{(p)})^2=x_p\ .
\ee
Keeping in mind that $x_{i}\equiv x_{j_\ell+1}=\mu_k^{-1}(\zeta_{\ell}^+)$ and $x_{i-1}=\mu_k^{-1}(\zeta_{\ell})$ one has  
\be
\mu_k^{-1}(\zeta^+_{\ell})-\mu_k^{-1}(\zeta_{\ell}) \geq 2\,\left(\gamma^2_{\ell+1}\mu_k^{-1}(\zeta^+_{\ell})-\gamma^2_{\ell}\mu_k^{-1}(\zeta_{\ell})\right)\,\left(\int^1_{\mu^{-1}_k(\zeta_{\ell}^+)} \mu_k([0,x])\,dx \right)^2 \ .
\ee
Since $\mu_k^{-1}(\zeta_{\ell})\leq \mu_k^{-1}(\zeta_{\ell}^+)$ the statement holds true.
\end{proof}

\section{Conclusions and perspectives}\label{sec:conclusions}

In this paper we have investigated the phase diagram of the multiscale Sherrington-Kirkpatrick model and proved some new features. First of all, we clarify the underlying synchronization property of the solution of the variational problem \eqref{eq:var_principle}. 
Following that, we computed the limiting overlap moments in the various multiscale measures in Theorem~\ref{thm:overlap_moments}, whose ordering is shown to reflect consistently the synchronization property mentioned earlier. In Theorem~\ref{thm:annealed} we exhibit a high-temperature, { or weak-coupling}, sufficient condition for an annealed solution to hold, which entails also the vanishing of all the overlap second moments. Theorem~\ref{thm:minRSB} instead pinpoints a {strong coupling regime} in which, even in absence of an external magnetic field, the solution of the variational principle \eqref{eq:var_principle} must present at least $r$ (i.e. the number of scales originally planted in the model) distinct points in its support.

\begin{figure}[h!]
    \centering
    \begin{tikzpicture}[scale=2, every node/.style={font=\small}]
    \draw[->, thick] (0,0) -- (4,0) node[right] {$\gamma^2$};

    \fill[red!10] (0,0.1) rectangle (1,-0.1);
    \draw[red, ultra thick] (0,0) -- (1,0);
    \fill[blue!10] (1,0.1) rectangle (3.95,-0.1);
    \draw[blue, ultra thick] (2.2,0) -- (3.95,0);
    \fill[shading = axis, left color=red!10, right color=blue!10] (1,0.1) rectangle (2.2,-0.1);
    \draw[purple, ultra thick] (1,0) -- (2.2,0);


    \draw[thick] (0,0.1) -- (0,-0.1);
    \draw[thick] (1,0.1) -- (1,-0.1);
    \draw[thick] (2.2,0.1) -- (2.2,-0.1);
    \node[right] at (0.0,-0.25) {$0$};
    \node[right] at (2.2,-0.25){$\gamma_1^2>\frac{1}{2\zeta^2_1}$};

    \draw[dashed] (0,0.6) -- (0,-0.6);
    \draw[dashed] (1,0.6) -- (1,-0.6) ;
    \node[left] at (1,-0.25) {$\gamma_r^2\leq\frac{1}{2}$};
    \draw[dashed] (2.2,0.6) -- (2.2,-0.6) ;
\end{tikzpicture}
    \caption{Synthetic representation of our findings.}
    \label{fig:phase_diag}
\end{figure}

{
\figurename~\ref{fig:phase_diag} summarizes our results until Theorem \ref{thm:minRSB}, in absence of external magnetic field. All $\gamma$'s have been collapsed on a 1-D line for the sake of visualization. The red part stands for the annealed region, where $\gamma_r^2\leq 1/2$. Afterwards there is an intermediate phase where a \emph{partial annealing} may take place, namely the conditional probabilities $\rho_{\mu,\ell}$ have all the mass in $0$ up to a certain $\ell$. In this shaded region, some of the $\gamma^2_\ell$'s can be greater than $1/(2\zeta_\ell^2)$, and some smaller. Following that, we have the strong coupling region highlighted in blue, where partial annealing together with finite RSB is forbidden by Theorem
\ref{thm:minRSB}.  We believe moreover that partial annealing is not possible in this region in any cases. We plan to return to this matter in a subsequent work.

Finally, in  Theorem~\ref{thm:plateau} we proved that, under suitable conditions, the CDF associated to a solution of the variational problem \eqref{eq:var_principle},   must have plateaus in correspondence of the values $\zeta_\ell$. In said Theorem we give a lower bound on the size of such plateaus, which also holds under a local (i.e.\ for every $\ell$) weak-coupling condition. The existence of at least one plateau is ensured only if 
all the mass is not concentrated at zero. This for sure happens when the external magnetic field is present, i.e. $\EE h^2>0$, as proved in Theorem \ref{thm:minRSB}, or when $\gamma_r^2>1/2$ as proved in Theorem \ref{thm:annealed}. 
We also note that for ``very strong'' couplings full replica symmetry breaking can occur without the Parisi measure manifesting gaps in the support.

Let us now discuss some possible perspectives on this model. It would be interesting to further investigate the shaded region. More precisely, we were able to prove a sufficient condition for an atom to detach from zero at a certain time scale, which is the content of Lemma \ref{lem:partial_annealing--}. It is thus natural to wonder if it is also necessary, thus leading to a condition similar to that for the total annealing of Theorem \ref{thm:annealed}. This in turn would give us a fine control on the possible partial annealing phenomenon. Furthermore, it is not clear if the condition $\gamma^2_{\ell+1}<1/(2\zeta_\ell)$ of Theorem \ref{thm:plateau} is fundamental, or it is just a technical artifact. Finally, in presence of external fields instead, it would be interesting to establish a type of de Almeida-Thouless condition for the multiscale SK model, namely, to check if there exists a range of the parameters where each or part of the conditional measures $\rho_{\mu^*,\ell}$ are Dirac deltas centered away from zero. In more physical terms, this would amount to check replica symmetry for each of the thermodynamic scales.}

To conclude, we plan to investigate the effect of the multiscale structure combined with other mean-field models, such as multispecies models \cite{barra2015multi,panchenko_multi-SK} where different types of synchronization (species and scales) must coexist.  

\section{Acknowledgements}
FC, PC and EM were supported by the EU H2020 ICT48 project Humane AI Net contract number 952026; by the Italian Extended Partnership PE01 - FAIR Future Artificial Intelligence Research - Proposal code PE00000013 under the MUR National Recovery and Resilience Plan; by the project PRIN22CONTUCCI, 2022B5LF52,  Boltzmann Machines beyond the independent identically distributed Paradigm: a Mathematical Physics Approach, CUP J53D23003690006.\\ 
DT was supported by project SERICS (PE00000014) under the MUR National Re-
covery and Resilience Plan funded by the European Union - NextGenerationEU; by the project PRIN22TANTARI, 20229T9EAT, Statistical Mechanics of Learning Machines: from algorithmic and information-theoretical limits to new biologically inspired paradigms, 
CUP J53D23003640001. The work is partially supported by GNFM (Indam). PC, EM, and DT are affiliated to GNFM-INdAM.

\bibliography{AHP_revised}

\end{document}